\pdfoutput=1

\newif\ifLCNS
\LCNSfalse

\newif\ifFull
\Fulltrue
\ifLCNS \Fullfalse \fi

\ifLCNS
\documentclass{llncs}
\else
\documentclass[11pt]{article}
\advance \topmargin by -\headheight
\advance \topmargin by -\headsep
\evensidemargin \oddsidemargin
\fi

\usepackage{bm}
\usepackage{graphicx}
\usepackage{verbatim}
\usepackage{algorithm}
\usepackage{algorithmic}
\usepackage{amssymb,amsmath,latexsym}
\usepackage{multirow}
\usepackage{array}
\usepackage{times}
\usepackage{url}
\usepackage{xcolor}
\usepackage{array}
\usepackage{hyperref}
\usepackage{wrapfig}

\ifLCNS
\renewcommand{\subparagraph}{}
\usepackage[small, compact]{titlesec}

\usepackage{enumitem1}
\setlist{nosep}

\fi

\usepackage{adjustbox}

\usepackage[numbers, sectionbib,sort]{natbib}
\ifLCNS  \fi

\makeatletter
\def\@begintheorem#1#2{\sl \trivlist \item[\hskip \labelsep{\bf #1\ #2:}]}
\def\@opargbegintheorem#1#2#3{\sl \trivlist
      \item[\hskip \labelsep{\bf #1\ #2\ (#3):}]}
\makeatother

\ifLCNS
\else
\newtheorem{theorem}{Theorem}[section]
\newtheorem{definition}[theorem]{Definition}
\newtheorem{lemma}[theorem]{Lemma}

\newenvironment{proof}{\noindent{\bf Proof:}}{\hspace*{\fill}\rule{6pt}{6pt}\bigskip} 
\fi

\newcommand{\shuffle}{Melbourne~}
\newcommand{\msgsize}{\mathsf{msgSize}}
\newcommand{\bucket}[1][]{\ifthenelse{\equal{#1}{}}{}{\mathsf{#1}\_}\mathsf{bucket}}
\newcommand{\m}[1]{\mathsf{#1}}
\newcommand{\Enc}{\mathsf{Enc}}
\newcommand{\Dec}{\mathsf{Dec}}
\newcommand{\key}{\mathsf{key}}
\newcommand{\dummy}{{\mathsf{dummy}}}
\newcommand{\seed}{{\mathsf{seed}}}
\newcommand{\loc}{{\mathsf{loc}}}
\newcommand{\adv}{{\mathcal{A}}}
\newcommand{\advB}{{\mathcal{B}}}
\newcommand{\me}{\mathrm{e}}
\newcommand{\getRange}{\m{getRange}}
\newcommand{\putRange}{\m{putRange}}
\newcommand{\getRangeDist}{\getRange\m{Dist}}
\newcommand{\putRangeDist}{\putRange\m{Dist}}
\newcommand{\proofInApp}[1]{\ifLCNS The proof is given in Appendix~\ref{#1}.\fi}
\newcommand{\shufflepass}{\m{shuffle\_pass}}
\newcommand{\optimshufflepass}{\m{optim\_shuffle\_pass}}
\newcommand{\setup}{\mathsf{Setup}}

\newcommand{\comments}[1]{}

\makeatletter

\newcommand{\Rmnum}[1]{\expandafter\@slowromancap\romannumeral #1@}
\newcommand*{\medcup}{\mathbin{\scalebox{1.5}{\ensuremath{\cup}}}}
\makeatother

\begin{document}

\title{\ifLCNS \large \fi The Melbourne Shuffle: 
                Improving Oblivious Storage in the Cloud}

\newcommand{\institutes}{$^1$Microsoft Research, $^2$University~of California, Irvine, $^3$Brown University}

\author{%
\ifLCNS \normalsize\fi
Olga Ohrimenko$^1$ \and
Michael T. Goodrich$^2$ \and 
Roberto Tamassia$^3$ \and 
Eli Upfal$^3$ 
\ifFull   \vspace{8pt}  \and \small\institutes \fi
}
\ifLCNS \institute{\institutes} \fi

\date{}

\maketitle

\pagestyle{plain}

\begin{abstract}
We present a simple, efficient, and secure
data-oblivious randomized shuffle algorithm.
This is the first secure data-oblivious shuffle that is not based on sorting. 
Our method can be used to improve previous oblivious
storage solutions for network-based outsourcing of data.
\end{abstract}

\section{Introduction}
\label{sec:intro}

One of the unmistakable recent trends
in networked computation and distributed information management is that of 
{\it cloud storage} (e.g., see~\cite{kk-ccs-10}), whereby users outsource 
data to external servers that manage and provide access to their data.
Such services relieve users from the burden of backing up and having
to maintain access to their data across multiple computing platforms, but,
in return, such services also introduce privacy concerns.
For instance, it is likely that cloud storage providers will want to
perform data mining on user data, and it is also possible that such
data will be subject to government searches.
Thus, there is a need for algorithmic solutions that preserve the
desirable properties of cloud storage while also providing privacy
protection for user data.

Of course, users can encrypt data they outsource to the cloud,
but this alone is not sufficient to achieve privacy protection,
because the data access patterns that users exhibit can 
reveal information about the content of their data
\ifFull
(e.g., see~\cite{cwwz-sclwa-10,DBLP:conf/ndss/IslamKK12}).
\else
(e.g., see~\cite{DBLP:conf/ndss/IslamKK12}).
\fi
Therefore, there has been considerable 
amount of recent research on algorithms for 
\emph{data-oblivious algorithms and storage}, which 
hide data access patterns for cloud-based
network data management solutions
(e.g., see~\cite{%
gm-paodor-11,gm-shuffle-12,gmot-orsew-11,prac_oram,logn_oram,%
ostrovsky_soda,scsl-orwcc-11,%
sss-tpor-11,path_oram,williams_single,wsc-bcomp-08}).
Such solutions typically work
by obfuscating a sequence of data accesses intended by a client 
by simulating it with the one
that appears indistinguishable from a random sequence of data accesses.
Often, such a simulation involves mixing the intended (real) accesses with
a sequence of random ``dummy'' accesses.
In addition,
so as to never access the same address twice (which would reveal
a correlation), such obscuring simulations also
involve continually moving items around in the server's memory space.
For this reason, the ``inner-loop'' computation required by such simulations is
a \emph{data-oblivious shuffling} operation, which moves
a set of items to random locations in fashion that disallows the server
to correlate the previous locations of items with their new locations.
This inner-loop process
requires putting items
in new locations that are independent of their old locations
while hiding the correlations between the two. 

The most common way this inner-loop
shuffling is implemented is, however, computationally expensive,
since it involves assigning
random (or pseudo-random) indices to items and then performing a data-oblivious 
sorting of these index-item pairs.
Examples of such
oblivious sorting algorithms
include Batcher's sorting network~\cite{b-snta-68}, which
requires  $O(n (\log n)^2)$ I/Os to sort data of size $n$,
or the AKS~\cite{aks} or Zig-zag sorting~\cite{g-zddos-14} networks,
which use 
$O(n \log n)$ I/Os, but with large constant factors that restrict
their practicality.
These algorithms are used in oblivious storage solutions by having
a client use the server as an external memory, with the I/Os directing the
client to issue commands to move
items from the server to the client's private memory and from the client's
private memory to the server.
Though these solutions achieve a desired
privacy level, they are expensive
in their (amortized) access overhead time and also in their monetary
cost when one considers a client outsourcing large volumes of data
and accessing it from a cloud server that charges per every data request.

In this paper, therefore, we are interested in algorithmic improvements
for oblivious storage solutions, 
in terms of their conceptual
complexity, constant factors, and monetary costs.
For instance, since cloud-storage servers typically charge users for each 
memory access but have fairly large bounds on the size of the messages
for such I/Os, we allow for messages to have modest sizes, such
as $O(\sqrt{n})$ for a storage of size $n$.
This necessarily also implies that the client has an equally
modest-sized private memory, 
in which to send and receive such messages (and also in which to perform
internal swaps of data items away from the prying eyes of the server).
Our goal in this research is to take advantage of such frameworks
to replace data-oblivious sorting with simple oblivious data shuffling
for the sake of providing simple, efficient, and cheap outsourced
data management.
Our framework, therefore, involves designing (or modifying) oblivious storage 
simulation algorithms where a client stores $n$ items at the server and
is allowed to issue a sequence of I/Os, each of which is 
a batch of reads and writes for the server's memory, for reasonable
assumptions on message size and private memory size.


\paragraph{Related Work.}

A \emph{shuffle} is an algorithm for rearranging an array to achieve a random
permutation of its elements.  Early shuffle methods were
motivated by the problem of shuffling a deck of cards.

\ifFull
Classic card shuffle methods (e.g.,
Knuth (or Fisher-Yates)~\cite{knuth-vol2},
the riffle shuffle~\cite{ad-shuffle-86},
Thorp shuffle~\cite{thorp1973}) are not data-oblivious,
however, as anyone observing card swaps or riffles (interleaving
two subdecks) of such methods can learn the final output permutation.
\else
Classic card shuffle methods (e.g.,
\cite{knuth-vol2,ad-shuffle-86,thorp1973}) are not data-oblivious,
however, as anyone observing card swaps or riffles (interleaving
two subdecks) of such methods can learn the final output permutation.
\fi
In ICALP~2012, Goodrich and Mitzenmacher~\cite{gm-shuffle-12} 
showed that one can, in fact, shuffle
a deck of~$n$ cards and guarantee that an observer cannot find
a particular card in the output permutation with probability better
than~$O(1/n)$.  However, this algorithm is not an effective shuffle 
for our purposes, since the
output permutations produced by the algorithm are not all equally likely
and there may be dependencies between large groups of cards that 
could be leaked.
Most other existing efficient data-oblivious shuffling methods assign random
values to the elements of the array and use a data-oblivious algorithm
to sort the array according to these values.

\paragraph{Our Oblivious Shuffling Results.}
Our \shuffle shuffle\ifFull\footnote{The name of our algorithm is inspired by a
\small{\href{{http://en.wikipedia.org/wiki/Melbourne_Shuffle}}{\texttt{``shuffle''}}
dance technique}.}
\else
~
\fi
algorithm
is instead the first data-oblivious shuffle method that
is not based on a data-oblivious sorting algorithm.
In Table~\ref{tbl:cmp_obl}, we compare the \shuffle shuffle,
showing that it outperforms
sorting-based shuffle methods.

\begin{table}[hbt]
\small
\begin{center}
  \caption{\label{tbl:cmp_obl}
 Comparison of data-oblivious sorting and shuffle
    algorithms over $n$ items.}
\ifLCNS    
\vspace{-10pt}
\fi
  \begin{tabular}{ >{\arraybackslash}m{\ifFull 3.8cm \else 3.5cm \fi} | >{\centering\arraybackslash}m{0.77cm} | >{\raggedleft\arraybackslash}m{1.5cm} | >{\raggedleft\arraybackslash}m{1.5cm} | >{\raggedleft\arraybackslash}m{1.5cm} | >{\raggedleft\arraybackslash}m{2cm} }
  & \parbox{0.6cm}{\adjustbox{angle=60,width=\width-1.1em}{~~Randomized}}  &
  {Private Memory}&
    {Message Size}&
      {External Memory} &
        {I/Os} \\
  \hline
  Batcher's network~\cite{b-snta-68} & & $O(1)$ & $O(1)$ & $O(n)$ & $O(n (\log n)^2)$ \\
  Batcher's network~\Rmnum{1} & & $O(\sqrt{n})$ & $O(\sqrt{n})$ & $O(n)$ & $O(\sqrt{n} (\log n)^2)$ \\
  Batcher's network~\Rmnum{2}~\cite{gm-paodor-11} & & $O(\sqrt{n})$ & $O(\sqrt[8]{n})$ & $O(n)$ & $O(n^{7/8})$ \\
  AKS~\cite{aks}, Zig-zag sort~\cite{g-zddos-14} & & $O(1)$ & $O(1)$& $O(n)$& $O(n \log n)$ \\
  Randomized shellsort~\cite{g-rsaso-10} & $\checkmark$ & $O(\sqrt{n})$ & $O(\sqrt{n})$ & $O(n)$ & $O(\sqrt{n} \log n)$ \\
  \hline
  \shuffle shuffle & $\checkmark$ & $O(\sqrt{n})$&  $O(\sqrt{n})$& $O(n)$ & $O(\sqrt{n})$ \\ 
 \shuffle shuffle ($c \geq 3$) & $\checkmark$ & $O(\sqrt[c]{n})$&  $O(\sqrt[c]{n})$& $O(n)$ & $O(c\sqrt[c]{n^{c-1}})$ \\ 
  \end{tabular}
\end{center}
\vspace{-17pt}
\end{table}

\paragraph{Improved Oblivious Storage.}
Oblivious storage and oblivious RAM (ORAM)
simulation solutions aim at minimizing the \emph{access overhead}, 
which is the amortized number of I/Os executed to perform a single
storage access request while keeping reasonable assumptions about the
size of private memory and of messages exchanged between the client
and the server\ifFull, e.g., sublinear in the size of outsourced
memory~$n$\fi. \ifFull In seminal work motivated by a software
protection application, \fi  Goldreich and
Ostrovsky~\cite{gold87,go-spsor-96} give two oblivious storage solutions
for a client with $O(1)$ private memory size:
the square root method with $O(\sqrt{n})$ overhead and the hierarchical method
with $O((\log n)^3)$ overhead. 
The hierarchical method was recently extended using techniques such as Bloom
filters~\cite{wsc-bcomp-08,williams_single} and cuckoo hash
tables~\cite{logn_oram,ostrovsky_soda,gm-paodor-11,pr-orr-20}.  E.g., the $\log n$-hierarchical solution of~\cite{logn_oram} uses
$O(\sqrt[d]{n})$ temporary memory and achieves $O(\log n)$ access
overhead, for~$d \ge 2$.  In~\cite{ostrovsky_soda} a similar method achieves $O({(\log
  n)}^2/ \log \log n)$ overhead with $O(1)$ private memory.  
\ifFull

\fi
All the above oblivious storage solutions rely on a periodic
data-oblivious shuffle of the server storage, a task done using
data-oblivious sorting. \ifFull This is the most expensive step of such solutions,
but since it happens only after a certain number of requests, it can
be amortized.  Otherwise, one can use techniques
of~\cite{gmot-orsew-11} to deamortize these solutions to bring the
worst case access overhead to be the same as the average case. \fi
Thus, we can use our \shuffle shuffle to implement the shuffle steps
of these algorithms.

Other oblivious storage solutions in~\cite{path_oram,scsl-orwcc-11,sss-tpor-11} allow the user to have~$o(n)$ private memory
and then by applying the same solution recursively on this private
memory bring it to $o(1)$ and adding a $\log n$ overhead.  For
example, Path ORAM~\cite{path_oram} uses $O(\log n)$ (stateful)
private memory and has $O((\log n)^2)$ access overhead.

As shown in~Table~\ref{tbl:oram}, oblivious storage solutions based
on our \shuffle shuffle are efficient and
practical.
\begin{table}[hbt]
\small
\begin{center}
  \caption{\label{tbl:oram}
 Comparison of oblivious storage solutions  for $n$ items.}
\ifLCNS
\vspace{-10pt}
\fi
  \begin{tabular}{ >{\arraybackslash}m{\ifFull 5.5cm \else 4.94cm \fi} | >{\raggedleft\arraybackslash}m{1.65cm} | >{\raggedleft\arraybackslash}m{1.65cm} | >{\raggedleft\arraybackslash}m{1.65cm} | >{\raggedleft\arraybackslash}m{1.65cm} }
  & {Private Memory}&
    {Message Size}&
      {External Memory} &
        {Access Overhead}\\
  \hline
  SquareRoot~\cite{go-spsor-96}  & $O(1)$ & $O(1)$ &  $O(n)$ & $O(\sqrt{n})$\\
  Path ORAM~\cite{path_oram} & $O(\log n)$& $O(\log n)$ & $O(n)$ & $O({(\log n)}^2)$ \\
  Bucket Hash Hierarchical~\cite{go-spsor-96}  & $O(1)$ & $O(1)$ &  $O(n \log n)$ & $O((\log n)^3)$\\
  Cuckoo Hash Hierarchical~\cite{logn_oram} ($d\ge2$)  & $O(\sqrt[d]{n})$ & $O(\sqrt[d]{n})$ &  $O(n)$ & $O(\log n)$\\
  \hline
  SquareRoot with \shuffle shuffle & $O(\sqrt{n})$&  $O(\sqrt{n})$& $O(n)$ & $O(1)$ \\ 
{Hierarchical with \shuffle shuffle} & $O(\sqrt[c]{n})$&  $O(\sqrt[c]{n})$& $O(n)$ & $O(c \log n)$ \\   
$(c \geq 3)$ & $O(\sqrt[c]{n} \log n)$&  $O(\sqrt[c]{n} \log n)$& $O(n)$ & $O(c)$ \\ 
  \end{tabular}
\end{center}
\vspace{-10pt}
\end{table}


\section{Preliminaries}

\ifFull \subsection{Cryptographic Primitives} \fi
\label{sec:crypto}

We analyze the security of cryptographic
primitives and our protocol in terms of the probability of success
for an adversary in
breaking them.
Let $k$ be a security parameter.
\ifFull
We consider a probabilistic adversary~$\adv$
whose running time is polynomial in~$k$.
\fi
We say that a scheme is secure
if for every probabilistic polynomial time (in $k$) (PPT)
adversary~$\adv$,
the probability of breaking
the scheme is at most
some \emph{negligible function} $\m{negl}(k)$, i.e., a function such that
$\m{negl}(k) < 1/|\m{poly}(k)|$ for every polynomial $\m{poly}(k)$.

\ifFull\paragraph{CPA Secure Encryption} \fi
We use a \emph{symmetric encryption scheme} $(\m{Enc}_{\key}, \m{Dec}_{\key})$
where $\key \leftarrow \{0,1\}^k$.
We require this scheme to be secure against the chosen-ciphertext
attack (CPA) for multiple messages~\cite{katz-lindell-07}.
\newcommand{\CPAsecure}{%
  During this attack an adversary~$\adv$ is allowed to make queries to
  oracles $\m{Enc}_\key$ and $\m{Dec}_\key$ on a polynomial number of
  sequences of $l$ messages of his choice.  After this ``warm-up''
  phase, $\adv$ comes up with two sequences, $M_0$ and $M_1$, of $l$
  messages and gives them to a challenger. The challenger secretly
  picks a bit~$b$ and calls $\m{Enc}_\key$ on each message of
  sequence~$M_b$.  Let $C$ be the sequence of ciphertexts that
  correspond to~$M_b$.  The challenger gives~$C$ to~$\adv$ who
  continues querying $\m{Enc}_{\key}$ and~ $\m{Dec}_{\key}$ on any
  sequence of ciphertexts except those in~$C$ for a polynomial number
  of times.  Finally, the adversary's task is to guess bit~$b$.  We
  call the above game \emph{Enc-IND-CPA} and say that~$\adv$ wins the
  game if he correctly guesses~$b$.  Then, $(\m{Enc}_{\key},
  \m{Dec}_{\key})$ is said to be secure if for all PPT adversaries,
  the probability of winning game Enc-IND-CPA is at most $1/2 +
  \m{negl}(k)$.  We omit using $\key$ when referring to $(\Enc,
  \Dec)$.  For an intuition behind Enc-IND-CPA secure encryption
  scheme, consider encrypting a message padded with a different random
  nonce each time it is encrypted. Hence, re-encryptions of the same
  plaintext look different with very high probability.
}
\ifFull \CPAsecure \else%
Informally, in the \emph{Enc-IND-CPA} security game, the adversary~$\adv$
picks two sequences of plaintexts and gives them to the challenger,
who encrypts one of them and returns the sequence of
ciphertexts, $C$, to the adversary.  $\adv$~is given full oracle
access to \ifFull the encryption and decryption algorithms\else
$\Enc_\key$ and limited oracle access to~$\Dec_\key$\fi,
where he cannot call $\Dec_\key$ on ciphertexts in~$C$.
We say that
\ifFull the encryption scheme \else($\Enc_\key, \Dec_\key)$~\fi
is secure if for every PPT adversary the probability
of guessing which plaintext
sequence was used to produce~$C$ is at most
$1/2+\m{negl}(k)$.
A formal definition is in the Appendix (Section~\ref{app:ccpasecure}).
\fi

\ifFull\paragraph{Pseudo-Random Permutation (PRP)}\fi 
Consider an array of $n$ elements that we wish to randomly rearrange
and let $D=[1,n]$ be the set of indices of~$A$.  We use a family of
\ifLCNS secure and \fi efficiently computable \emph{pseudo-random
  permutations} (PRPs) $\Pi_\seed: D \rightarrow D$, keyed using a
\ifFull $\seed$ from the set $\m{Seeds}(\Pi) = \{0,1\}^k$\else $k$-bit
seed\fi~\cite{katz-lindell-07}.
\newcommand{\PRPs}{%
  In order to pick a permutation, one picks a random $\seed$ from
  $\m{Seeds}(\Pi)$ and stores it privately.  Hence, whenever we refer
  to choosing a permutation, we refer to picking a new $\seed \in
  \m{Seeds}(\Pi)$.  Once the seed is fixed, we can evaluate $\Pi_\seed$ on a
  given index $x \in D$ via~$\Pi_\seed(x)$.

  The security of a family of PRPs is defined by comparing the
  behavior of a PPT adversary when he is given a truly random
  permutation versus a pseudo-random permutation picked using a random
  seed.  Formally, let $\mathcal{R}$ be the set of all permutations
  over the domain~$D$. A family of PRPs $\Pi$ is secure if for every
  probabilistic polynomial time adversary~$\mathcal{A}$, the
  probability of distinguishing between $\m{r} \xleftarrow{\$}
  \mathcal{R}$ and $\Pi_\seed$, where $\seed \xleftarrow{\$}
  \m{Seeds(\Pi)}$, is $1/2 + \m{negl}(k)$. 
%
} 
\ifFull \PRPs \else The security of a pseudo-random permutation is
formally defined in Appendix~\ref{app:prp}.  \fi

Given an array~$A$ of $n$ (key,value) pairs $(x,v)$ where $x \in [1,n]$,
we denote the permutation $\pi$ of $A$ as $B=\pi(A)$,
where $\pi = \Pi_\seed$ and
$B[x]=A[\pi(x)]$, $\forall x \in [1,n]$.
We will use the same notation when $A$ and $B$ are encrypted.
We refer to the original permutation of $A$,
as permutation~$\pi_0$. For~$A$, sorted using~$x$,
$\pi_0$ is the identity.

\ifFull \subsection{Storage Model}\fi
\label{sec:storage_model}
We consider a cloud storage model where a client stores a
dataset at a server while keeping a small amount of data in private
memory. For simplicity, we assume that the dataset is an array of
elements of equal size.  The client encrypts each element
and stores the elements at the server according to a \ifFull pseudo-random
permutation\else PRP\fi. The encryption key and the seed of the
\ifFull permutation \else PRP \fi are
kept private by the client and are not revealed to the server.

\ifFull\paragraph{Client Private Memory} We assume that the client
has access to a small private memory, $M$, which is comprised of permanent
storage and scratch space.  The permanent storage includes the
encryption key and the current seed of the permutation the client is
using, which together is of size~$O(1)$.  The rest of $M$ is used as a
scratch space while performing operations on the remote storage and is
not needed in between operations.  We require the size of the scratch
space to be sublinear in~$n$. Depending on the algorithm, we will use
a private storage of size $\sqrt[c]{n} \log n$ or $\sqrt[c]{n}$,
for an arbitrary integer $c \ge 2$.

Since the client can store a small number of elements at a time, we
assume that he does not try to request from the server more than he
can fit and process in~$M$.  Let the \emph{message size}, denoted
$\msgsize$, be the maximum elements that can be exchanged by the
client and server in one operation. We have that $\msgsize$ should be
less than the size of the scratch space.

\fi

\ifFull\paragraph{Server Memory}
The server supports the following
operations on an array~$S$.
\begin{itemize}
\item $\mathsf{get}(S, \loc)$: return element stored at a location $\loc$ in $S$.
\item $\mathsf{put}(S, \loc,e)$: put element $e$ to a location $\loc$ in $S$.
\item $\getRange(S, \loc,\ell)$: return an array $a$ with elements
at locations $\loc, \ldots, \loc+\ell-1$ in $S$, where $\ell \le \msgsize$.
\item $\putRange(S, \loc,a)$: write elements in array $a$
to locations $\loc, \ldots, \loc+ |a|-1$ in~$S$, where $|a| \le \msgsize$.
\item $\getRangeDist(S, \langle \loc_1, \ldots, \loc_c, \rangle ,
\langle \ell_1, \ldots, \ell_c, \rangle)$:
return an array $a$ with elements
at locations $\loc_i, \ldots,$ $\loc_i+\ell_i-1$ in $S$,$\forall i\in [1,c]$,
where $\sum \ell_i \le \msgsize$.
\item $\putRangeDist(S, \langle \loc_1, \ldots, \loc_c\rangle,
\langle a_1, \ldots, a_c\rangle)$: write elements in each array $a_i$
to locations $\loc_i, \ldots,$ $\loc_i+ |a_i|-1$ in~$S$, where $\sum |a_i| \le \msgsize$.
\end{itemize}
\else%
The server supports a standard set of operations on an array~$S$:
$\getRange(S, \loc,\ell)$ returns the elements at locations in range
$\loc, \ldots, \loc+\ell-1$;
$\putRange(S, \loc,a)$ writes the elements of array $a$
to locations $\loc, \ldots, \loc+ |a|-1$;
$\putRangeDist(S, \langle \loc_1, \ldots, \loc_c\rangle,$
$\langle a_1, \ldots, a_c\rangle)$~is~a~generalization of~$\putRange$
that can write several arrays
to non-sequential locations.
\fi
\ifFull Note that the number of elements in $\getRange$ and
$\putRange$ is limited by the maximum number of elements that
can be exchanged between the client and the server in one operation.
\fi

\ifFull
We assume that the server can perform operations $\mathsf{get}$ and
$\mathsf{put}$ in constant time and operations $\getRange$,
$\putRange$, $\getRangeDist$ and $\putRangeDist$
in time proportional to the number of elements
read or written,
but
each operation takes one~I/O.
\else
We assume the server performs the above operations in time
proportional to the number of elements read or written, but
each operation takes one~I/O.
\fi

\ifFull

\begin{definition}[Metadata]
\label{def:metadata}
The name of the array~$S$,
its size, location~$i$,~$l$, and the size of~$a$
are referred to as the \emph{metadata}
of a $\getRange$ or a $\putRange$ call.
Similarly for $\getRangeDist$ and $\putRangeDist$,
locations $\langle \loc_1, \ldots, \loc_c\rangle$,
$\langle \ell_1, \ldots, \ell_c, \rangle$ and sizes
of $a_1, \ldots, a_c$ are referred as metadata as well.
\end{definition}

\fi

\comments{
\subsection{The Balls-and-Bins Model}
We briefly review the analysis of the
well-known balls and bins model and
guide the reader to~\cite{} for more
information.
Consider $n$ balls thrown in $m$ bins,
where each ball is placed in a bin
independently and uniformly at random.
Using Chernoff bounds, one can
show that in the case when $m = n$
with probability at least~$1/n^c$, there is
no more than $c \log n$ balls in any bin.
While in the case when $m = n^{1/c}$,
with probability at least~$XXX$ there is
no more than $n^{1- 1/c}$ balls in any bin.
}


\section{Oblivious Shuffle Model}

In this section, we introduce a formal model for the oblivious shuffle
of an array.

\ifFull \subsection{Model} \fi
\label{sec:model}

\ifLCNS \vspace{-8pt} \fi

\begin{definition}[Shuffle]
We define a \emph{shuffle}~$\mathcal{S}$ as a pair of
algorithms $(\mathsf{Setup}, \mathsf{Shuffle})$, as follows.
\ifFull
\begin{itemize}
\else
\begin{itemize}[noitemsep]
\fi
\item $(s, S) \leftarrow \mathsf{Setup}(1^k)$
  Given security parameter $k$, run the key generation algorithm for a
  symmetric encryption scheme $(\m{Enc},\m{Dec})$ and store the key in
  secret state~$s$.  Also, allocate an auxiliary datastore~$S$.
\item $(\m{Enc}(\pi(A)), \alpha) \leftarrow \mathsf{Shuffle}(s, S,
  A, \pi)$
  Given secret state $s$,~auxiliary data store $S$,
  an array input~$A$,
  and a permutation $\pi$, return (1)~the
  encryption of the permutation of~$A$ according to~$\pi$;
  (2)~a transcript~$\alpha$ of the operations
  that transform~$\m{Enc}(A)$ to~$\m{Enc(\pi(A))}$ using auxiliary
  space~$S$.
\end{itemize}
Transcript~$\alpha$ is a sequence of $l$ (request, response) pairs
  $\langle (r_1, g_1), \ldots,(r_l, g_l)\rangle $ that capture
  the evolution of the datastore 
  via intermediate states~$S_1, S_2, \ldots, S_{l+1}$. 
  An invariant on each intermediate state is to store
  an encryption of some permutation of~$A$ along with
  any auxiliary data.
  For example~$S_1$
  contains~$\m{Enc}(A)$ and~$S_l$ contains~$\m{Enc}(\pi(A))$.
  Setting $S_1\leftarrow \{\m{Enc}(A), S\}$, $s_0 \leftarrow s$, $g_0 \leftarrow \bot$
  define the relationship between $r_i$ and $g_i$~as:
  \ifLCNS
  \vspace{-15pt}
  \fi
  \begin{align*}
    \langle ~ (s_i, r_i) \leftarrow \mathsf{GenRequest}(s_{i-1}, g_{i-1}), ~~
    (S_{i+1},g_{i}) \leftarrow \mathsf{GenResponse}(S_i, r_i) ~ \rangle .
  \end{align*}
  \ifLCNS
  \vspace{-15pt}
  \fi
  
 \noindent
Operations $\mathsf{GenRequest}$ and $\mathsf{GenResponse}$ generate
  a request~$r_i$ and a corresponding response~$g_{i}$ and
  are defined as follows:
\ifFull
\begin{itemize}
\else
\begin{itemize}[noitemsep]
\fi
\item $(s_i, r_i) \leftarrow \mathsf{GenRequest}(s_{i-1}, g_{i-1})$ Perform a
  computation based on a substructure of $S_{i-1}$, $g_{i-1}$, and
  generate next request to~$S_i$,~$r_i$.
\item $(S_{i+1}, g_{i}) \leftarrow \mathsf{GenResponse}(S_i, r_i)$
  Generate the response to request $r_i$ on~$S_i$: $S_{i+1}$ is the
  datastore $S_i$ updated according to $r_i$ and $g_{i}$ is the
  response to~$r_i$ with respect to~$S_i$.  For example, if $r_i$ is a
  get request, then $S_{i+1}=S_i$ and $g_{i}$ is the requested
  item. Also, if $r_i$ is a put request, then $g_{i}$ is empty.
\end{itemize}
The private state~$s$ is updated if needed after every request.
\end{definition}
\ifLCNS \vspace{-4pt} \fi

\noindent

\ifFull
\begin{figure}[hbt!]
  \vspace{-12pt}
\begin{center}
\else
  \begin{wrapfigure}{r}{2.5in}
    \leavevmode\hspace*{-6pt}
\fi
\ifLCNS
\vspace{-10pt}
\fi

\includegraphics[scale=\ifFull 0.45 \else 0.33 \fi ]{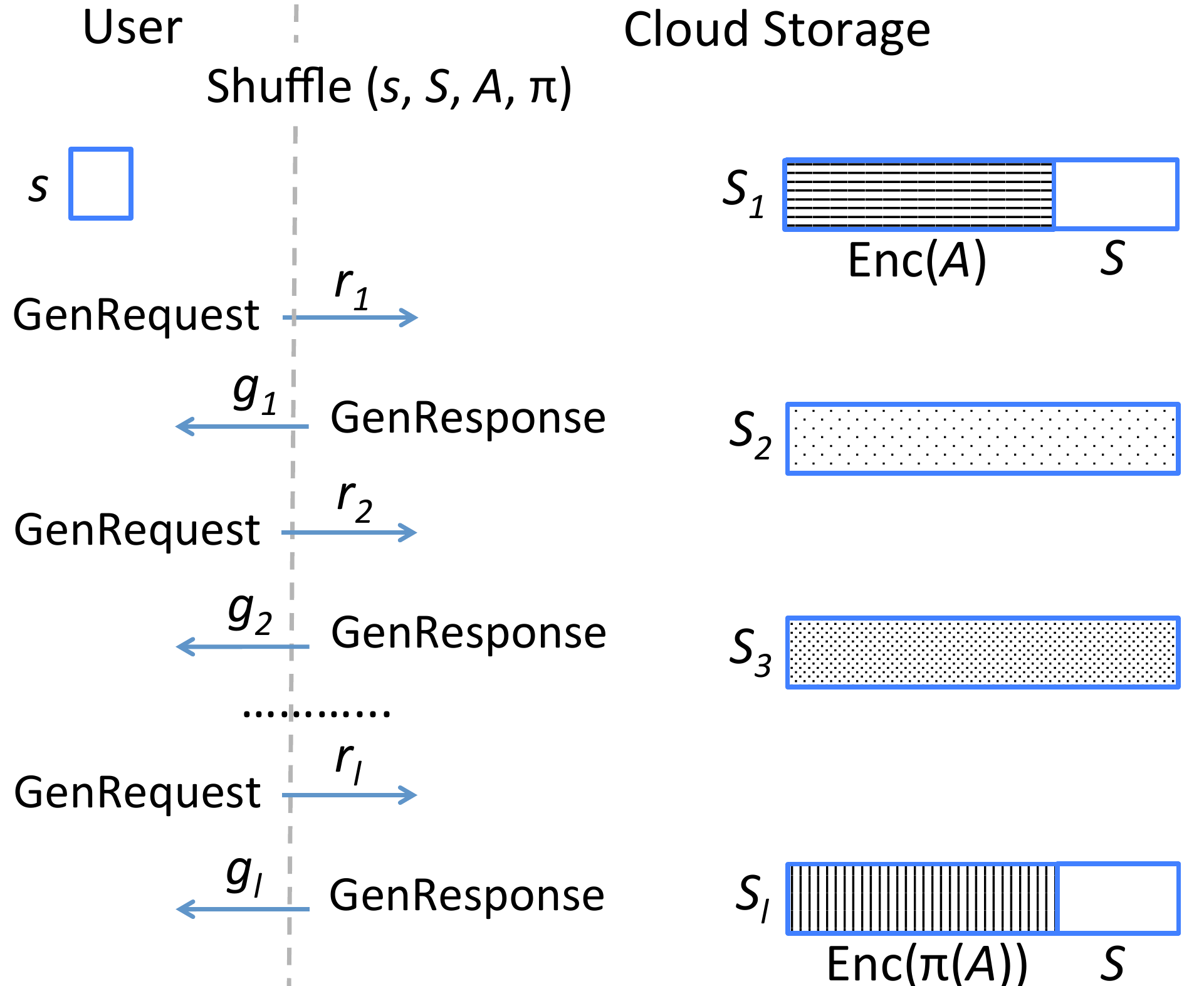}
\ifLCNS
\vspace{-14pt}
\fi
\caption{\label{fig:protocol}%
  Illustration of a shuffle $\mathcal{S}$ executed by the client and
  the cloud storage server (Section~\ref{sec:model}).  }
\ifFull
\end{center}
\end{figure}
\else
\vspace{-10pt}
  \end{wrapfigure}
\fi
In our cloud storage model, a shuffle~$\mathcal{S}$
is a distributed computation executed by the user and the server.
The user runs \ifFull
the $\mathsf{Setup}$ algorithm 
\else 
$\mathsf{Setup}$
\fi
to generate
the encryption key
and requests the server to allocate some space. 
He then runs \ifFull the $\mathsf{Shuffle}$ algorithm \else $\mathsf{Shuffle}$ \fi by accessing~$S$
through the server, that is, issuing requests to the server using~$\m{GenRequest}$.
The set of possible requests is defined by the storage
model supported by the server. 
\ifFull In our case this set is
$\{\mathsf{get},$ $\mathsf{put},$ $\getRange,$ $\putRange$, $\getRangeDist$, $\putRangeDist\}$
(see~Section~\ref{sec:storage_model}).
\fi
For every request~$r_i$, the server executes~$\mathsf{GenResponse}$,
locally updating~$S$ for put requests
and returning to the user the queried items for get requests.
(See Figure~\ref{fig:protocol} for an illustration.)

\ifFull
\subsection{Security} \fi

We capture the security of a shuffle $\mathcal{S}$ against a curious
server in the cloud storage model as a game, \emph{Shuffle-IND},
between~$\mathcal{S}$ and~a probabilistic
polynomial-time bounded (PPT) adversary~$\mathcal{A}$.
In this game, the inputs and outputs of~$\mathcal{S}$
that are revealed
to the server in the cloud storage model
are also revealed to~$\mathcal{A}$.
However, the secret state~$s$ kept by the client, any updates to it and
computations inside of~$\m{GenRequest}$ are kept private,
since in the cloud model they are also hidden and happen on the user side.

The game starts with $S$ running
$\m{Setup}$ once, allocating at the server space
to be used in subsequent computations.
$\mathcal{A}$ then tries to ``learn''
how~$\mathcal{S}$ performs the shuffle on a
sequence of $m_1$~input arrays and permutations picked by~$\mathcal{A}$.
Based on what~$\mathcal{A}$ learns, she picks
two challenges $(A_0, \tau_0)$ and $(A_1, \tau_1)$
each consisting of a data array to be permuted using
a corresponding permutation.
$\mathcal{S}$~secretly picks one pair and
performs the shuffle according to it.
The adversary is then allowed to observe~$\mathcal{S}$
shuffling another sequence of $m_2$ (input, permutation) pairs, also
picked by~$\mathcal{A}$.
Finally,~$\mathcal{A}$ has to guess which challenge
pair (input, permutation) $\mathcal{S}$ picked to shuffle.
Note that at any time, $\mathcal{A}$
can ask $\mathcal{S}$ to perform a shuffle
on any combination of $A_0$ or $A_1$ and permutations~$\tau_0$~or~$\tau_1$.

\ifFull We now give a formal definition of the game.  \else A formal
definition of the Shuffle-IND game is given in the Appendix
(Section~\ref{app:shufflesecure}).  \fi

\newcommand{\Shufflesecure}{%
\begin{definition}[Shuffle-IND]
\label{def:ind-cpa}
Let $A$ be an input array of size $n$ picked by
a PPT adversary $\mathcal{A}$.
$\mathcal{A}$ and~$\mathcal{S}$ engage in the following
game.
\begin{description}
\item[$\mathcal{S}$:] $(s,S) \leftarrow \setup(1^k)$.
\item for $j \in \{1, \ldots, m_1\}$, where $m_1$ is $\m{poly}(k)$:
   \begin{description} 
       \item $\mathcal{A}$: Pick array~$B_j$ and a permutation $\rho_j$.
       \item $\mathcal{S}$: Execute $(O_j, \alpha) \leftarrow \mathsf{Shuffle}(s, S, B_j, \rho_j)$.
       Reveal $O_j$ and~$\alpha$ to $\adv$.
     \end{description} 
  \item $\mathcal{A}$: Pick $(A_0, \tau_0)$ and $(A_1, \tau_1)$ of the same length.
   \item $\mathcal{S}$: Pick a secret bit $b$ and execute
 $(O, \alpha) \leftarrow \mathsf{Shuffle}(s, S, A_b, \tau_b)$. Reveal~$O$ and $\alpha$ to~$\adv$.
   \item for $j \in \{m_1+1, \ldots, m_1+m_2\}$, where $m_2$ is $\m{poly}(k)$:
   \begin{description}
      \item $\mathcal{A}$: Pick an encrypted array~$B_j$ and a permutation $\rho_j$.
       \item $\mathcal{S}$: Execute
        $(O_j, \alpha) \leftarrow \mathsf{Shuffle}(s, S, B_j, \rho_j)$. 
        Reveal $O_j$ and~$\alpha$ to $\adv$.
   \end{description}
   \item $\mathcal{A}$: Output bit $b'$.
   \item The adversary wins the game if $b=b'$.
   \end{description}   
\end{definition}
}

\ifFull \Shufflesecure \fi

\ifFull Using the Shuffle-IND game, we now define an \emph{oblivious
  shuffle}. \fi

\ifLCNS \vspace{-4pt} \fi
\begin{definition}[Oblivious Shuffle]
\label{def:obl-shuffle}
Let $k$ be the security parameter
and $n$ be a polynomial in $k$.
$\mathcal{S}$ is an oblivious shuffle over $n$~items
if for every \ifFull probabilistic adversary $\mathcal{A}$ running in time
polynomial in~$k$\else PPT adversary\fi,
the probability of winning the Shuffle-IND game\ifFull, $\Pr[b=b'] $, satisfies
$$\Pr[b=b'] \le \frac{1}{2} + \m{negl}(k).$$\else~is at most $1/2 +
\m{negl}(k)$. \fi
\end{definition}
\ifLCNS \vspace{-4pt} \fi

\ifFull \subsection{Performance} 

We measure the performance of an oblivious shuffle of an array of~$n$
items~$A$ using
the following parameters:
\begin{description}
\item \emph{Number of Requests}: the number of calls the user and the server
make to each other while performing the shuffle.
In the protocol, the number of requests is expressed using~$l$,
the length~of the transcript~$\alpha$.
This parameter measures the efficiency of a shuffle algorithm.
\item \emph{Message Size:} the maximum number of items
that can be sent between the user and the server in a
single operation,
i.e., the number of items sent in a request~$r_i$ or~a response~$g_i$.
\item \emph{User Private Memory:} the size of user's private memory~$s$ measured in terms
of the number of items of~$A$ that can be stored temporarily
plus the space required
to store an encryption key and a seed for a pseudo-random permutation.
The key and the seed are stored in the stateful part of~$s$ which is of a constant size,
while the space required to store the items is the scratch space
required only during the execution of the shuffle algorithm and is erased
afterwards.
\item \emph{User Computation}: the amount of computation the user performs
during the~$\m{Shuffle}$ algorithm, i.e., computation inside of~$\m{GenRequest}(s_i,g_{i-1})$.
\item \emph{Server Storage}: additional space required at the
server besides storing $n$~encrypted items of~$A$. This is captured by the datastore~$S$
in the protocol.
\item \emph{Server Computation}:
the computation performed by the server during~$\m{GenResponse}(S_i, r_i)$.
\end{description}

\else

\fi

\ifFull

In the cloud storage model, we wish to devise efficient shuffle solutions
that consider realistic
assumptions about the user.
For example, a user private memory of size~$n$
leads to a trivial and efficient solution where the user can download
the data from the server, shuffle and re-encrypt the items,
and send them back.
Instead we wish to construct solutions that add a small overhead
over a non-oblivious solution in terms
of the number of requests and client computation, but
assume private memory and message size sublinear in~$n$.
We will also assume that the server and the user can exchange
more than one item in one message or call.
The size of each message is limited by the size of user's
memory since this is how many items she can process at a time.

Finally, we consider current cloud storage providers that store
user data and can efficiently retrieve
and write small amounts of data as requested by the user.
Such storage providers charge their users
according to a pay-per-use
model. Hence, we wish to limit space requirements
of our solutions.

\fi

\section{The \shuffle Shuffle}
\label{sec:shuffle}

In this section, we present a basic version of our \shuffle shuffle
algorithm. An optimized version is given in the next
section. 
\ifFull
The basic \shuffle shuffle uses private memory and messages of size $O(\sqrt{n} \log n)$,
$O(n \log n)$ server storage 
and processes in $O(\sqrt{n})$ requests.
The optimized \shuffle shuffle has a smaller
message and memory overhead,
and requires constant number more accesses to the server.
In particular, it relies on private memory and messages
of size $O(\sqrt{n})$, $O(n)$ server storage
and $O(\sqrt{n})$ accesses.
\else
Background and intuition for  \shuffle shuffle are given in
the Appendix (Section~\ref{app:intuition}).
\fi

\newcommand{\Shuffleintiuition}{%

An important ingredient of our solution is probabilistic encryption.
Everything stored at the server is encrypted and every time
an item is read from the server, the user decrypts it, re-encrypts it
and writes it back. Since we use CPA-secure encryption,
the ciphertexts produced for the same item always look
different and, hence, the server, aka the adversary,
cannot tell whether the ciphertexts correspond to the same item or not.

The goal of our oblivious shuffle is to reveal
to the adversary only information that she would
expect to see in a random permutation
with very high probability.
For example, even for a secret permutation
picked uniformly at random,
the adversary can guess with probability~$1/n$
that the first element of the input array of size~$n$ appears
in some location~$i$ of the output permutation.
Continuing with this intuition, suppose we split the input array of size~$n$
into~$\sqrt{n}$ buckets where every bucket has
$\sqrt{n}$ items, and similarly for the output
permutation.
In this case, using the analysis of the balls-and-bins model,
the adversary can guess that with high probability,
each bucket in the output permutation
has  $O(\log n)$ elements
from any particular bucket of the input array.

We build on the observation above
and move elements from input buckets
to output buckets
by imitating the balls-and-bins process.
That is, if the size of the
input bucket is the same as the number
of output buckets, we place $O(\log n)$
elements of every input bucket in every output bucket.
If the number of elements in a bucket is much larger than
the number of output buckets, i.e., the number
of output buckets is~$n^{1/c'}$ while input bucket has
$n^{1/c}$ items, for constants $c$ and $c'$ s.t. $c' > c$, then
we move~$O(n^{1/c - 1/c'})$ items
to every output bucket.

The reader may have noticed that in the first example
above, elements of an input bucket of size $\sqrt{n}$
are placed in $\sqrt{n}$ output buckets in batches of~$O(\log n)$
items. What are the additional items?
These additional items are referred to as \emph{dummy items}.
A dummy item is a real item with a fake key and some nonce
value such that the size of the dummy and real item are equal.
Moreover, since all the data is re-encrypted every time
it is written to the server, the server cannot tell
which items in a batch are real and which are dummy.

}

\ifFull \Shuffleintiuition \fi

\ifFull \subsection{Overview} \else
\paragraph{Overview}
\fi

We assume that each element in the input array, $A$, is a key-value
pair $(x,v)$ for every $x\in D$.
The algorithm proceeds
in two phases: \emph{distribution} and \emph{clean-up}.
For each phase,
the data store $S$ is split in several logical subparts: $I$, $T$ and $O$.
$I$~is an array containing $n$~encrypted items of the input~$A$
permuted according to some permutation~$\pi_0$ (initially,
$\pi_0$ is the identity).
$T$ is an encrypted temporary array used during the shuffle;
finally, after the shuffle is done $O$ contains the output of the shuffle,
i.e., re-encrypted items
of~$I$ permuted according to~$\pi$.
If the shuffle needs to be executed again,
the user sets $I \leftarrow O$ and $\pi_0 \leftarrow \pi$.
We further divide each subpart of~$S$
in buckets of equal size. The number of buckets
and how it effects the runtime of the algorithm
will be determined later.

During the distribution phase items of every bucket of~$I$ along with
some dummy items are re-encrypted and distributed equally among
buckets of~$T$. Here, the distribution of item $(x,v)$ is done
according to its final location $\pi(x)$ in $O$.
After the distribution phase the intermediate array~$T$ contains
real and dummy items. Moreover, the items
appear in correct buckets but not in correct positions within each bucket.
The clean-up phase remedies this by reading one bucket at
a time, removing dummy items, distributing the real items
correctly within the bucket and writing the bucket to $O$.

The distribution phase alone cannot produce every possible permutation
since the number of items sent from a bucket of $I$ to a
bucket of $T$ is limited. E.g., the identity permutation cannot be
achieved.  To rectify this, we execute two shuffle passes.  First, for
a permutation~$\pi_1$ picked uniformly at random and then for the
desired permutation~$\pi$.  Although this framework still allows
failures, our algorithm can produce every permutation,
failing with very small probability independent of the desired
permutation~$\pi$.

\ifFull \subsection{Algorithm} \else
\paragraph{Algorithm} \fi
\label{sec:algorithm}

The complete shuffle algorithm~$\m{shuffle}(I, \pi, O)$
is shown in~Algorithm~\ref{alg:shuffle-alg}
\ifFull
where $I$ is the encryption of the input array~$A$,
$\pi$ is the desired permutation
and the last argument is the output array where the
algorithm is expected to put an encryption of~$\pi(A)$.
We omit the $\setup$ from the discussion since
it is trivial: the client simply runs a key generation
to setup a secure encryption scheme
and a seed generator for pseudo random
permutations.

\else .
\fi
The algorithm makes two calls to
$\m{shuffle\_pass}$
\ifFull (Algorithm~\ref{alg:shuffle-pass}),
\else
(Algorithm~\ref{alg:shuffle-pass} in the Appendix, Section~\ref{app:shufflepass}),
\fi
first for a random permutation $\pi_1$
and then for the desired permutation~$\pi$.
We proceed with the description
of $\m{shuffle\_pass}(I, T, \rho, O)$
where $I$ and~$O$ are defined as in
$\m{shuffle}$, $T$ is a temporary
array and $\rho$ is the desired permutation.
We use the convention of giving arrays $I$, $T$ and~$O$
as inputs to the shuffle pass algorithm for the ease of explanation.
In the cloud storage scenario that we consider here, one simply specifies the location
where these arrays are stored remotely\ifFull, e.g., the name of a file and a
location within it\fi.
Given an input array of size $n$,
this method has messages
and client's private memory of size $O(\sqrt{n} \log n)$
and server memory of size $O(n \log n)$.
These user and server memory requirements are temporary
and are reduced to
$O(1)$ and $n$,~respectively, when the shuffle is finished.
As mentioned before, method $\shufflepass$ is split into a
distribution phase and a clean-up phase.

\begin{algorithm}[t!]
\caption{\label{alg:shuffle-alg}%
  The complete \shuffle shuffle algorithm, $\mathsf{shuffle}(I, \pi, O)$, where
  the user can read and store in private memory~$M$ up to $\sqrt{n} \times p \log n$~elements,
  $p \ge \me$.}
$\bm{I}$: array of $n$ encrypted elements $(x,v)$;
$\bm{\pi}$:  permutation;
 $\bm{O}$: permutation of $I$ according to $\pi$, where every element is re-encrypted.
\begin{algorithmic}[1]
    \STATE Let $\pi_1$ be a random permutation
    \STATE Let $T$ be an empty array of size $n \times p \log n$ stored remotely 
    \STATE $\m{shuffle\_pass}(I, T, \pi_1, O)$
    \STATE $I \leftarrow O$
    \STATE $ \m{shuffle\_pass}(I, T, \pi, O)$
 \end{algorithmic}
\end{algorithm}

\newcommand{\Shufflepass}{%
\begin{algorithm}[hbtp]
\caption{\label{alg:shuffle-pass}%
  Single pass $\mathsf{shuffle\_pass}(I, T, \rho, O)$ of the \shuffle shuffle algorithm, where
  the user can read and store in private memory~$M$ upto $\sqrt{n} \times p \log n$~elements.}
$\bm{I}$: array of $n$ encrypted elements $(x,v)$;
$\bm{T}$: auxiliary array that fits $n \times p \log n$ encrypted elements,
where $p$ is a constant that is a parameter of the algorithm;
$\bm{\rho}$: permutation;
 $\bm{O}$: permutation of $I$ according to $\rho$, where every element
is re-encrypted.
\begin{algorithmic}[1]
    \STATE $\m{max\_elems} \leftarrow p \log n$
    \STATE $\m{num\_buckets} \leftarrow \sqrt{n}$
    \STATE \COMMENT{\emph{Distribution phase: distribute elements of~$I$ into~$T$}}
    \FOR[read buckets of $I$]{$\mathsf{id}_I \in \{0, \ldots, \m{num\_buckets}-1\}$}\label{line:simple-dist-for-start}
    	\STATE $\bucket_M \leftarrow \mathsf{getRange}(I,\mathsf{id}_I\times \sqrt{n}, \sqrt{n})$ \label{line:simple-getRange1}
	\STATE $\bucket[rev]_M \leftarrow \mathsf{empty\_map}()$ \COMMENT{Reverse map of bucket ids in $T$ to elements}~
	\FOR[Assign elements their bucket ids in $T$]{$e \in \bucket_M$}
	         \STATE $(x,v) \leftarrow \Dec(e)$
	         \STATE $\mathsf{id}_T \leftarrow  \lfloor \rho(x)/\sqrt{n} \rfloor$  \COMMENT{Bucket id of element $(x,v)$ in $T$ according to its location in $O$}
	         \STATE $\bucket[rev]_M[\mathsf{id}_T].\mathsf{add}(\Enc(x,v))$  \COMMENT{Collect elements of same bucket} \label{line:simple-enc}
	\ENDFOR
	 \STATE \COMMENT {Can be done via a single $\putRangeDist$ for $\sqrt{n}$ batches of size $\m{max\_elems}$}
	\FOR[Distribute $\bucket_M$ in buckets of $T$]{$\mathsf{id}_T \in \{0, \ldots,\m{num\_buckets}-1\}$} \label{line:simple-for-start}
		\IF{$\mathsf{size}(\bucket[rev]_M[\mathsf{id}_T] ) > \mathsf{max\_elems}$} \label{line:simple-fail-if-start}
			\STATE \textbf{fail}  \label{line:fail} \COMMENT{$\rho$ moves more than $p\log n$ elements from a bucket of $I$ to a bucket of $T$}
		\ENDIF \label{line:simple-fail-if-end}
		\STATE \COMMENT{Hide how many real elements go to $T$'s buckets by padding with encrypted dummies}~
		\STATE $\bucket[rev]_M[\mathsf{id}_T] \leftarrow \m{dummy\_pad}(\bucket[rev]_M[\mathsf{id}_T], \m{max\_elems})$ \label{line:simple-pad}
		
		\STATE \COMMENT{Write a batch of $\mathsf{max\_elems}$ from every bucket of $I$ to every bucket of~$T$}~
		\STATE $\mathsf{putRange}(T,\mathsf{id}_T\times \sqrt{n} \times \mathsf{max\_elems} +\mathsf{max\_elems}\times \mathsf{id}_I, \bucket[rev]_M[\m{id}_T])$  \label{line:simple-putRange1} 
	\ENDFOR	  \label{line:simple-for-end}
    \ENDFOR \label{line:simple-dist-for-end}
    \STATE \COMMENT{\emph{Clean-up phase: clean $T$ and write the result to~$O$}}
    \FOR[read buckets of $T$]{$\mathsf{id}_T \in \{0, \ldots, \m{num\_buckets} -1\}$}    
    	\STATE $\bucket_M \leftarrow
		\mathsf{getRange}(T,\mathsf{id}_T\times \sqrt{n} \times\mathsf{max\_elems}, \sqrt{n} \times\mathsf{max\_elems})$ \label{line:simple-getRange2}
	\STATE \COMMENT{Decrypt the bucket, remove $\dummy$,
              sort real elements using $\rho$ and re-encrypt}
	\STATE  $\bucket_M \leftarrow \m{clean}(\bucket_M)$ \label{line:simple-clean}
	\STATE \COMMENT{The distribution phase guarantees that $\bucket_M$ contains exactly $\sqrt{n}$ elements}~
	\STATE $\mathsf{putRange}(O,\mathsf{id}_T\times \sqrt{n},\bucket_M)$	 \label{line:simple-putRange2} 
    \ENDFOR
 \end{algorithmic}
\end{algorithm}
}

\ifFull \Shufflepass \fi

\paragraph{Distribution Phase}
The distribution phase of method $\shufflepass$ \ifFull
(Algorithm~\ref{alg:shuffle-pass})\fi, shown in
Figure~\ref{fig:distr-phase}, imitates throwing balls into bins by
putting elements from every bucket of~$I$ to every bucket of~$T$
according to the permutation $\rho$.  In particular, a batch of $p
\log n$ encrypted elements from every bucket of $I$ is put in every
bucket of $T$ ($\bucket[rev][\m{id}_T]$ in the pseudo-code).  Here,
$p$ is a constant and is determined in the analysis.

\ifFull
\begin{figure}[t!]
\else
\begin{figure}[hbt]
\fi
\begin{center}
\includegraphics[scale=\ifFull 0.45 \else 0.33 \fi]{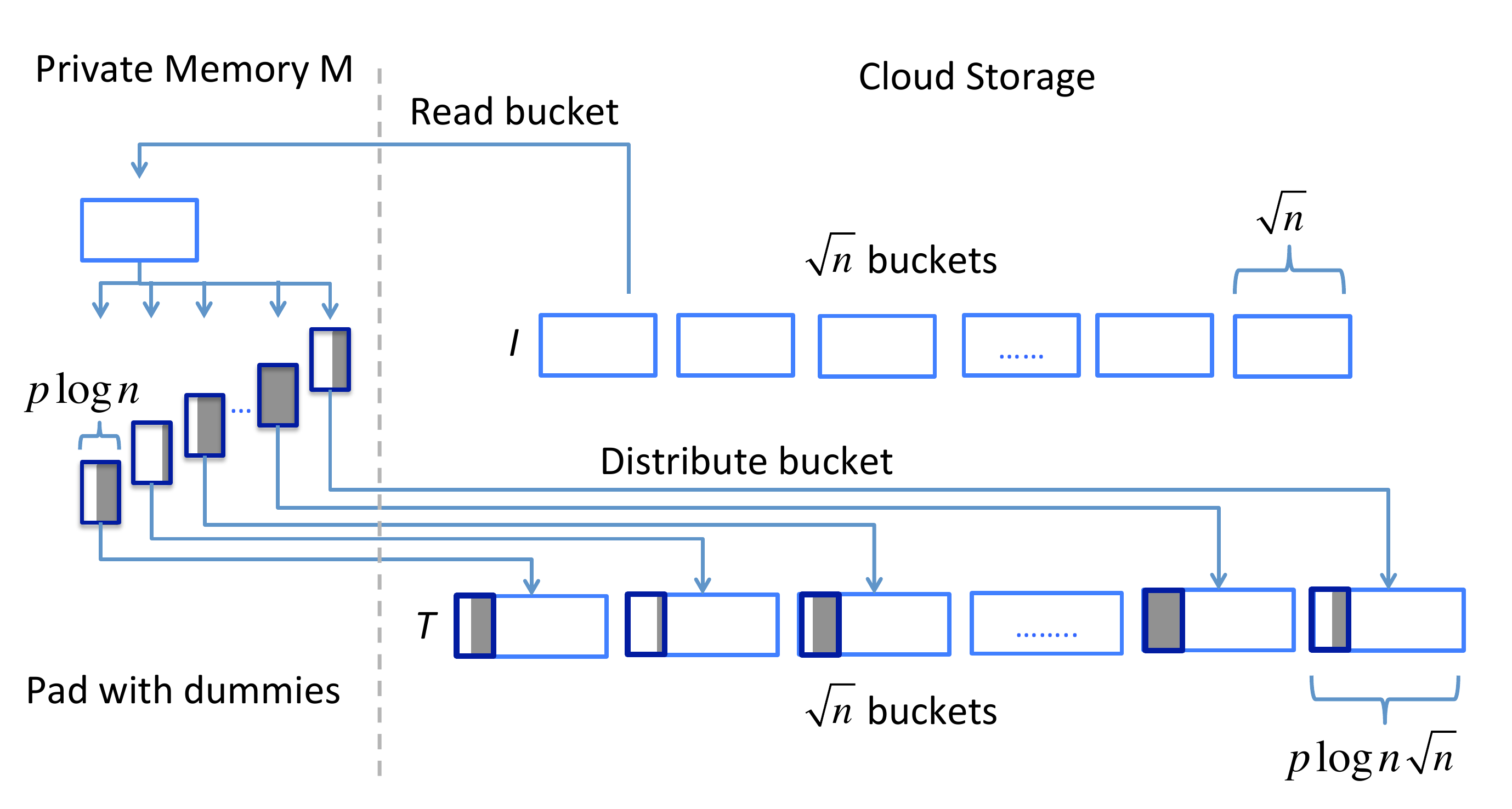}
\vspace{-12pt}
\caption{Illustration of the distribution phase of
  $\m{shuffle\_pass}$ (Algorithm~\ref{alg:shuffle-pass}).
  Shadowed regions represent dummy values added
  to pad each batch to the size of~$p \log n$. 
  The batches are encrypted, hence, one cannot tell
  where and how many dummy values there are in each batch.}
\label{fig:distr-phase}
\end{center}
\vspace{-5pt}
\end{figure}

Each batch contains real and dummy elements.
The first batch is filled in with real elements $(x,v)$
that would go to the first bucket in $O$ according
to~$\rho$, i.e., the elements for which
$\lfloor \rho(x)/\sqrt{n} \rfloor = 0$. Similarly for every
other batch.
Since a bucket of $I$ contains only $\sqrt{n}$ elements and we put
$\sqrt{n} \times p \log n$ elements in total in all buckets in~$T$, most batches will
have less than~$p \log n$ elements.
We pad such batches with dummy elements
to hide where and how many elements of $I$'s bucket are placed in~$T$~(line~\ref{line:simple-pad}).
Note that a batch is re-encrypted before it is written to $T$, completely hiding the content
and making it impossible to recognize where dummy or real elements are (lines~\ref{line:simple-enc} and~\ref{line:simple-pad}).
If according to $\rho$ more than $p \log n$ elements are mapped from
a bucket of $I$ to a bucket of $T$, the algorithm fails (line~\ref{line:fail}). We later consider
what happens in case of a failure.
\ifFull We note that $\sqrt{n}$ calls to~$\putRange$
in the loop in lines~\ref{line:simple-for-start}-\ref{line:simple-for-end}
is for the ease of explanation only.
These calls can be substituted by a single call~$\putRangeDist$,
putting $\sqrt{n} \times p \log n$ elements all at once.
Hence, for every bucket read from~$I$, there is only \emph{one}
corresponding write to~$T$.
\else
We note that all $\sqrt{n} \times p \log n$ elements from batches of a single bucket can be written
to~$T$ using a single call to~$\putRangeDist$.
\fi

\paragraph{Clean-up Phase}
The distribution phase leaves $T$ with two problems:
first, though the elements are in correct buckets according to $\rho$ they are not in
the correct locations inside the buckets, and second, $T$ contains dummy elements.
To remedy these problems, the clean-up phase
in~Algorithm~\ref{alg:shuffle-pass},
illustrated in Figure~\ref{fig:cleanup-phase},
proceeds by reading buckets of~$T$ of size
$\sqrt{n} \times p \log n$ and writing in their place buckets
of size $\sqrt{n}$. 

\ifFull  \begin{figure}[bt]  \else \begin{figure}[hbt] \fi
\begin{center}
\includegraphics[scale=\ifFull 0.45 \else 0.33 \fi]{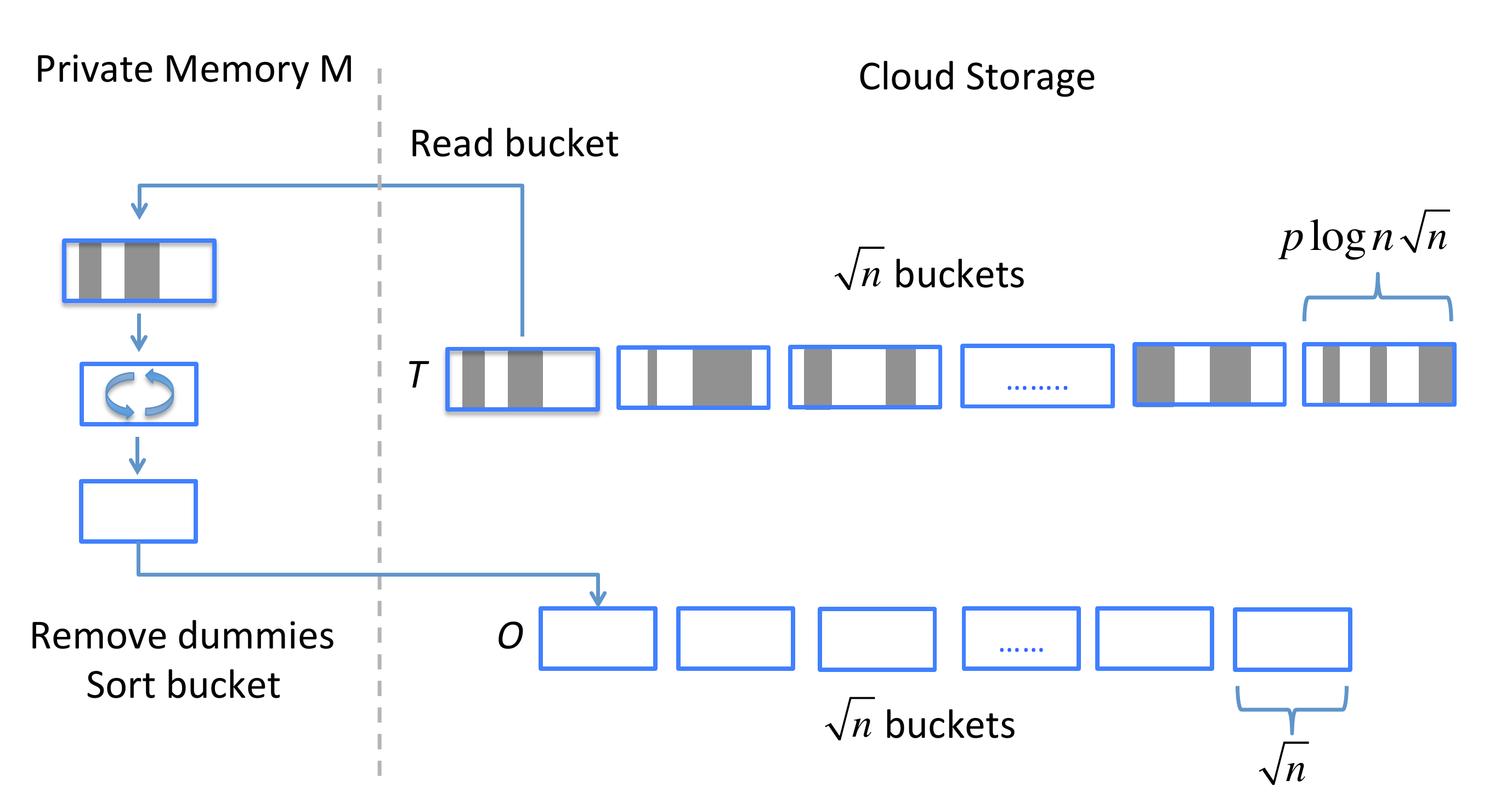}
\vspace{-12pt}

\caption{An illustration of the clean-up phase of $\m{shuffle\_pass}$  (Algorithm~\ref{alg:shuffle-pass}). Shadowed regions represent dummy values
that are removed during the clean-up phase.}
\label{fig:cleanup-phase}
\end{center}
\vspace{-14pt}

\end{figure}

When processing each bucket, the algorithm removes
dummy elements, sorts
the remaining content of every bucket according to their final location in~$O$~(line~\ref{line:simple-clean}).
It is important to note that each written bucket contains exactly $\sqrt{n}$ elements
before it is being written back. This follows from the fact that
elements were distributed to buckets according to the permutation~$\rho$
and the algorithm failed in the distribution phase for those~$\rho$
that would have resulted in more than~$\sqrt{n}$ elements in each bucket.
\comments{
We note that for the $\sqrt{n}$ case the first part of~Algorithm~\ref{fig:sort-pass}
is applicable, since we can read buckets of size $\sqrt{n}$ and sort
them in private memory~$M$.}

\paragraph{Performance} 
The performance of the \shuffle shuffle is summarized in the following
theorem.
\proofInApp{app:performancebasic}

\begin{theorem} \label{theorem:performance}
Given an input array of size $n$, 
the Melbourne shuffle (Algorithm~\ref{alg:shuffle-alg})
executes $O(\sqrt{n})$ operations, each exchanging a message of  size
$O(\sqrt{n} \log n)$, between
a user with private memory
of size  $O(\sqrt{n} \log n)$ and a server
with storage of size $O(n \log n)$. Also, the user and server
perform $O(n \log n)$ work.
\end{theorem}
\newcommand{\ProofPerformancebasic}{%
\begin{proof}
We first note that $\sqrt{n}$ calls to~$\putRange$
in the loop in lines~\ref{line:simple-for-start}-\ref{line:simple-for-end}
is for the ease of explanation only.
These calls can be substituted by a single call~$\putRangeDist$,
putting $\sqrt{n} \times p \log n$ elements all at once.

A single shuffle pass in Algorithm~\ref{alg:shuffle-pass}
requires $2\sqrt{n}$ calls to $\getRange$,
$\sqrt{n}$ calls to $\putRangeDist$,
and $\sqrt{n}$ calls to $\putRange$,
assuming the user and the server can exchange up to $\sqrt{n} \times p \log n$
elements in a single request.
The shuffle in Algorithm~\ref{alg:shuffle-alg}
requires $8\sqrt{n}$ requests in total since it makes
2 calls to the shuffle pass procedure.
The private memory required at the user to perform the shuffle is $\sqrt{n} \times p \log n$.
The required server's memory is $n \times p \log n$.
However, this overhead is temporary
since the increase in memory happens only during the shuffle pass
and is reduced to $n$ when the shuffle is finished.
Similarly for the user,
the memory of size~$\sqrt{n} \times p \log n$ is required only
during the shuffle.
We note that the total computation for the user and the server is $O(n \log n)$.
\end{proof}
}

\ifFull \ProofPerformancebasic \fi

\subsection{Security Analysis}
In this section, we show that the \shuffle shuffle
(Algorithm~\ref{alg:shuffle-alg}) is oblivious for every
permutation~$\pi$ with high probability. \ifLCNS Complete proofs are in the
Appendices~\ref{app:bbperms}--\ref{app:shuffle-alg}. \fi

\begin{definition}
\label{def:p-perm}
Let $A$ be an array of $n$ elements such that every $x \in [1,n]$ is
at location~$\pi_0(x)$ in~$A$. Let~$B$ be an array
that  stores a permutation $\pi$ of elements in~$A$,
i.e., $B = \pi(A)$.
Split~$A$ and $B$ in $\sqrt{n}$ buckets of equal size
and fix a constant $p \ge \me$.
Let $\pi$ be a permutation on $n$~elements
where every bucket of $B$ contains at most $p \log n$
elements of every bucket of $A$.
We refer to the set of all such permutations
as $P(\pi_0)$.
\end{definition}

\begin{lemma}
\label{lemma:bb-perms}
The size of set $P(\pi_0)$
is $(1-\m{negl}(n)) \times n!$, for every
permutation~$\pi_0$.
\end{lemma}
\newcommand{\ProofBBperms}{%
\begin{proof}
Let $\pi$ be a random permutation
from all possible $n!$ permutations.
We consider the relationship between
the input array~$A$ and a permutation
of~$A$, $B = \pi(A)$.
We start by splitting $A$ and $B$ in buckets of size $\sqrt{n}$
and numbering buckets from 1 to $\sqrt{n}$ using their order
in each array.
The analysis below estimates
how many permutations
can be constructed by restricting
the maximum number of elements from a bucket of $A$
appearing in any bucket of $B$ to $p \log n$.

Let $X_a^b$ be a random variable
that measures the number of elements
from $a$th bucket of~$A$ present in~$b$th bucket
of~$B$. The mean value of $X_a^b$ is 1,
since we are distributing $\sqrt{n}$ elements of~$a$
among~$\sqrt{n}$ buckets of~$B$.
Although $X_a^b$, for~$1 \le a,b \le \sqrt{n}$,
variables are dependent
between each other, we can use the Poisson
Approximation~\cite[Chapter 5.4]{mu-pcrap-05}
and instead work with~$n$ independent Poisson random variables $Y_{a}^b$
with mean 1.

Given $n$ variables $Y_{a}^b$ we are interested
in bounding the probability of the event that
there is no $a$ and $b$ such that $Y_{a}^b \ge p \log n$.
For a specific $a$ and $b$ it is:
\[
\Pr[Y_{a}^b \ge p \log n] \le \frac{1}{\me} \left(\frac{\me}{p \log n}\right)^{p \log n}.
\]
Using union bound, the probability
that at least one of the $Y_{a}^b$s
is greater than~$p \log n$ is at most
\[
 n \frac{1}{\me} \left(\frac{\me}{p \log n}\right)^{p \log n}.
\]
Since $Y_{a}^b$s are a Poisson
approximation of the variables
$X_a^b$, the probability
that at least one of the $X_{a}^b$s
is greater than~$p \log n$ is at most
\[
 2 n \frac{1}{\me} \left(\frac{\me}{p \log n}\right)^{p \log n}.
 \]
Setting $p \ge \me$ we get
\[
 2 n \frac{1}{\me} \left(\frac{\me}{p \log n}\right)^{p \log n}  \le \frac{2n}{(\log n)^{p \log n}} =
 \frac{2n}{n^{p \log \log n}} = 2\m{negl}(n) = \m{negl}(n).
\]
\ifLCNS \qed \fi
\end{proof}
}
\ifFull \ProofBBperms \fi

\begin{lemma}
\label{lemma:correctness}
Let $\pi_0$ be the initial permutation of $n$ elements
in~the input array~$I$.
Method $\shufflepass$ (Algorithm~\ref{alg:shuffle-pass}) succeeds for
all permutations $\rho \in P(\pi_0)$.
\end{lemma}
\newcommand{\ProofCorrectness}{%
\begin{proof}
The algorithm allocates elements of $I$ in $O$
according to $\rho$ by first putting them
into corrects buckets (lines~\ref{line:simple-dist-for-start}--\ref{line:simple-dist-for-end}) and then sorting
every bucket using $\rho$ (line~\ref{line:simple-clean}).
By construction the algorithm fails for
any permutation~$\rho$ that requires
more than $p \log n$ elements from a bucket of $I$
mapped to buckets of $O$~(line~\ref{line:simple-fail-if-start}--\ref{line:simple-fail-if-end}).
\ifLCNS \qed \fi
\end{proof}
}
\ifFull \ProofCorrectness \fi

\begin{lemma}
\label{lemma:two-pass-fail}
Method $\m{shuffle}(I, \pi, O)$ (Algorithm~\ref{alg:shuffle-alg})
is a randomized shuffle algorithm
that succeeds with very high probability.
\end{lemma}
\newcommand{\ProofPassfail}{%
\begin{proof}
Let $\pi_0$ be the initial permutation
of the input array~$I$.
Algorithm~\ref{alg:shuffle-alg} makes two calls to shuffle
pass which
succeeds for all possible permutations
except for $1/n^{\Omega(\log \log n)}$ fraction of them~(Lemmas~\ref{lemma:bb-perms}
and~\ref{lemma:correctness}).
The first shuffle pass is executed for permutation $\pi_1$
on an input permuted according to some permutation~$\pi_0$. 
Since $\pi_1$ is picked using internal random
coins, the probability of the shuffle pass failing is independent of $\pi_0$ and
is bounded by~$1/n^{\Omega(\log \log n)}$.
If the first shuffle pass did not fail,
the shuffle pass is executed second time
with input permutation~$\pi$.
The second shuffle is executed on
the input array that is permuted according
to a random permutation $\pi_1 \in P(\pi_0)$.
The second pass does not fail iff $\pi_1 \in P(\pi)$.
Hence, Algorithm~\ref{alg:shuffle-alg} fails
if $\pi_1 \not\in P(\pi_0)$
or  $\pi_1 \not\in P(\pi)$.
By Lemma~\ref{lemma:bb-perms},
the probability of either of these events is negligible in $n$,
hence the basic \shuffle shuffle succeeds with very high probability
for any $\pi_0$ and~$\pi$.
\ifLCNS \qed \fi
\end{proof}
}
\ifFull \ProofPassfail \fi

We show that method $\shufflepass$ (Algorithm~\ref{alg:shuffle-pass})
is oblivious by mapping it to the Oblivious Shuffle Model in
Section~\ref{sec:model}, extracting the corresponding transcript and
showing that the transcript reveals no information about the
underlying permutation if the encryption scheme is CPA secure~(see
Section~\ref{sec:crypto}).

Method
$\shufflepass$ (Algorithm~\ref{alg:shuffle-pass}) corresponds to
$\m{GenRequest}$ in the model and calls to~$\getRange$ and~$\putRange$
trigger calls to $\m{GenResponse}$ at the server.  We do not describe
$\m{GenResponse}$ since it depends on the implementation details of
the remote storage provider.  We are only interested in the fact that
it uses server's state $S$ to store and maintain arrays $I, T$ and~$O$.
The transcript~$\alpha$ of the shuffle execution is defined as
follows.  The request~$r_i$ is either $\getRange(S, i, l)$ or
$\putRange(S, i, a)$, e.g., in line~\ref{line:simple-putRange2} $(S, x,
a)$ is $(O, \m{id}_T\times \sqrt{n},\m{bucket}_M$).  The response
$g_i$ to $\getRange$ is an array~$a$, e.g., $a$ contains $\sqrt{n}$
elements in line~\ref{line:simple-getRange1} and is stored
in~$\m{bucket}_M$. The response to $\putRange$ is empty. 
We first analyze \ifFull the metadata~(Definition~\ref{def:metadata})
that corresponds to \else the metadata (i.e., arguments not based on content) of \fi
every request between the client and the server, and show that, unless
the algorithm fails, they depend on the size of the input only, and
are independent from the input array and the desired permutation.
Hence, we obtain that the \shuffle shuffle is a data independent
shuffle algorithm. We finally show that if the content exchanged is
encrypted, as it is in method $\shufflepass$
(Algorithm~\ref{alg:shuffle-pass}), the \shuffle shuffle
(Algorithm~\ref{alg:shuffle-alg}) is oblivious.

\newcommand{\ProofIndep}{%
\begin{lemma}
\label{lemma:data-indep}
The metadata of requests exchanged between
the client and the server in method $\shufflepass$
(Algorithm~\ref{alg:shuffle-pass}) is independent of permutation~$\pi_0$
of the input~$I$
and output permutation~$\rho \in P(\pi_0)$,
and depends only on~$n$.
\end{lemma}
\begin{proof}
Let~$\alpha$ be a sequence of
(request, response) pairs exchanged between
the client and the server, denoted as $(r_i, g_i)$,
where~$r_i$ is either a $\getRange$
or $\putRange$ and~$g_i$ is $\bucket_M$ for~$\getRange$
and empty for~$\getRange$.
The sequence $\alpha$ can
be further split in $\sqrt{n}$ $(\getRange, \putRange)$ calls
that correspond to distribution phase and
$\sqrt{n}$ $(\getRange, \putRange)$ during
the clean-up phase.

The metadata of~$\putRange$
is the name of the array, the location within the array
and how many elements should be read. In the algorithm
these correspond to reading an array of size~$n$
sequentially
in buckets of size~$\sqrt{n}$ (distribution phase, line~\ref{line:simple-getRange1})
and~$\sqrt{n} \times p \log n$ (clean-up phase, line~\ref{line:simple-getRange2}).
These data depend only on~$n$ and $p$.
The metadata for a $\putRange$ call
consists of the array to be accessed, the location where
to put the data and the size of the data to be
written.
In the algorithm, first calls to~$\putRange$
place~$\sqrt{n}$ batches of size $p \log n$
in the temporary array to locations
that depend on the input bucket that has been
read using~$\getRange$~(line~\ref{line:simple-putRange1}).
These
locations are deterministic
since $\getRange$ simply scans the input array.
The second sequence of calls to $\putRange$
happens during the clean-up phase when buckets of size~$\sqrt{n}$
are written sequentially to the output array~(line~\ref{line:simple-putRange2}).
These calls are also deterministic.
It is also easy to show that
the transcript where a sequence of $\sqrt{n}$ calls to~$\putRange$
in lines~\ref{line:simple-for-start}-\ref{line:simple-for-end}
is substituted with a single call to $\putRangeDist$
is also deterministic.
\ifLCNS \qed \fi
\end{proof}
}
\ifFull \ProofIndep \fi

\newcommand{\ProofPass}{%
\begin{lemma}
\label{lemma:shuffle-pass}
Let $\pi_0$ be the initial permutation of $n$ elements
in~the input array~$I$ and let~$\rho$ be a permutation
from the set $P(\pi_0)$.
Method $\m{shuffle\_pass}(I, T, \rho, O)$
(Algorithm~\ref{alg:shuffle-pass})
is an oblivious
shuffle according to Definition~\ref{def:obl-shuffle}.
\end{lemma}
\begin{proof}
We showed in~Lemma~\ref{lemma:correctness} that
Algorithm~\ref{alg:shuffle-pass} succeeds
for all $\rho \in P(\pi_0)$.
We also showed that all metadata
in the transcript that is revealed to
the adversary~$\mathcal{A}$
in the Shuffle-IND game in~Definition~\ref{def:ind-cpa}
after every call to $\m{Shuffle}$
is independent of data content and hence
can be determined based only on~$n$
and is the same for any choice of input
and output.
The data content exchanged in each call
does depend on the data, however,
it is always encrypted.
We show that the security of the shuffle
depends on the security of the
underlying encryption scheme.

To the contrary, we assume that there is a PPT adversary~$\adv$
that can distinguish with a non-negligible
advantage
two permutations~$\tau_0$ and $\tau_1$
by observing the transcript of one of them.
Hence, this adversary can win Shuffle-IND game.
We show that if~$\adv$ exists then
we can construct an adversary~$\mathcal{B}$
who can use $\adv$ to win Enc-IND-CPA game
with a non-negligible advantage,
which would break our assumption about
the encryption scheme.
We recall that in Enc-IND-CPA
game~$\mathcal{B}$
has access to
oracles~$\m{Enc}$ and~$\m{Dec}$ and
can encrypt and decrypt sequences of messages of his choice,
except asking for the decryption of the challenge
ciphertext.

We construct the adversary~$\advB$
as follows. $\advB$ does not need to
run $\m{Setup}$ since
he has oracle access to~$\m{Enc}$
and~$\m{Dec}$.
$\adv$ starts making calls to $\m{Shuffle}$
on chosen pairs of input arrays and permutations.
$\advB$ imitates the shuffle by responding
with the encrypted permutation and a transcript $\alpha$, that
he can produce himself.
He continues doing so until~$\adv$ comes up
with a challenge of two (input, permutation) pairs
$(A_0, \tau_0)$ and $(A_1,\tau_1)$.
$\advB$~first creates a static transcript
that will be the same for both permutations.
He then, extracts all the calls to be made
to $\m{Enc}$ into two sequences: one
that corresponds to $(A_0,\tau_0)$
and one to $(A_1,\tau_1)$.
He gives these two sequences of elements
to be encrypted to his own challenger.
The challenger picks one sequence at random,
encrypts all its plaintexts (i.e., elements) one by one and gives
the result to~$\advB$.
$\advB$~combines the ciphertexts with the metadata,
that depend only on~$n$,
to create a valid transcript~$\alpha$ and sends
it to $\adv$.
He continues, responding to~$\m{Shuffle}$
requests from $\adv$
until $\adv$ outputs his guess~$b$
for the pair $(A_b,\tau_b)$.
$\advB$ outputs $b$ as his guess
for which
sequence of messages his challenger picked.
Since $\adv$'s advantage in winning the
game is non-negligible so is~$\advB$'s.
\ifLCNS \qed \fi
\end{proof}
}
\vspace{-6pt}
\ifFull \ProofPass \fi

\begin{theorem}
\label{thm:shuffle}
The Melbourne Shuffle (Algorithm~\ref{alg:shuffle-alg})
is a randomized shuffle algorithm
that succeeds with very high probability
and is data-oblivious according to Definition~\ref{def:obl-shuffle}.
\end{theorem}
\newcommand{\ProofShuffle}{%
\begin{proof}
Algorithm~\ref{alg:shuffle-alg} makes two
calls to method
$\m{shuffle\_pass}$, which is oblivious by Lemma~\ref{lemma:shuffle-pass}.
If Algorithm~\ref{alg:shuffle-alg}
is not oblivious, then there is a PPT adversary
$\adv$ who can distinguish
$\m{shuffle}$ of two permutations.
If $\adv$ exists, we can build an adversary~$\advB$
who can break the security of the underlying $\m{shuffle\_pass}$
using~$\adv$.
Whenever $\adv$ makes a call to $\m{Shuffle}$
for a permutation $\pi$,
$\advB$ first picks a random permutation~$\pi_1$
and calls $\m{shuffle\_pass}$ on $\pi_1$.
It then uses the output of this call and $\pi$
to make another call to~$\m{shuffle\_pass}$,
this time returning the output to~$\adv$.
This continues until
$\adv$ comes up with two challenge
pairs (input, permutation)
$(A_0, \tau_0)$ and $(A_1,\tau_1)$.
$\advB$ first picks a random permutation
and runs shuffle-pass on it.
He then gives his own challenger
the output of this shuffle along with $(A_0,\tau_0)$ and $(A_1,\tau_1)$.
The challenger returns to $\advB$
the shuffle according to~$(A_b,\tau_b)$, keeping $b$ secret.
$\advB$ forwards what he receives to $\adv$.
$\advB$ continues replying shuffle requests of~$\adv$
as he did before until $\adv$
makes a guess for $b$.
$\advB$ outputs this guess as his own.
\ifLCNS \qed \fi
\end{proof}

Note that method $\m{shuffle\_pass}$ outputs $\m{fail}$ for
some permutations. Whenever it does so,
$\advB$ sends this information to~$\adv$.
However, as showed in Lemma~\ref{lemma:two-pass-fail}
this happens with negligible probability for any pair of input
and output permutations and reveals nothing to the adversary
about the output permutation since the failure is due to secret
random bits.
}
\vspace{-8pt}
\ifFull \ProofShuffle \fi



\section{The Optimized \shuffle Shuffle}
\label{sec:optim_shuffle}

In this section we present an optimized version of the \shuffle shuffle
that has smaller memory requirements on the memory of the user and the server,
and succeeds with higher probability
than the basic version of the previous section.  \ifFull The framework
of the optimized version is similar: we first re-randomize the input,
i.e., shuffle it according to a random permutation $\pi'$ and then
shuffle $\pi'$ towards the desired permutation~$\pi$.  \fi The main
difference with the basic version lies in the shuffle pass.

\ifFull
\subsection{Algorithm} \else
\paragraph{Algorithm}
\fi

As in the basic version, the shuffle pass splits the input array $I$
and the output array $O$ in consequent buckets of size $\sqrt{n}$.
For auxiliary storage we use two temporary arrays $T_1$ and $T_2$ of
size $p_1 n$ and $p_2 n$, respectively, where $p_1, p_2 > 1$ are
constants to be determined in the analysis.  We split $T_1$ and $T_2$
in buckets of size $p_1\sqrt{n}$ and~$p_2\sqrt{n}$, respectively.  The
shuffle pass proceeds with two distribution phases, instead of one for
the basic version, followed by a single clean-up phase in the end.
Detailed pseudo-code is given in Algorithms~\ref{alg:optim-shuffle-alg}--\ref{alg:optim-distr-phase2}
\ifLCNS (Appendix~\ref{app:optimized}). \fi

\newcommand{\CodeOptimShuffle}{%
\begin{algorithm}[h]
\caption{\label{alg:optim-shuffle-alg}%
  The optimized \shuffle shuffle algorithm, $\mathsf{optim\_shuffle}(I, \pi, O)$, where
  the user can read and store in private memory~$M$ up to $ p_2\sqrt{n}$~elements.}
$\bm{I}$: an array of $n$ encrypted elements $(x,v)$;
$\bm{\pi}$: the desired permutation for elements of~$I$;
 $\bm{O}$: the output array containing the permutation of $I$ according to $\pi$, where every element~$(x,v)$
is re-encrypted.\\
\begin{algorithmic}[1]
    \STATE Let $\pi_1$ be a random permutation
    \STATE Let $T_1$ and $T_2$ be two empty arrays stored remotely of size $p_1n$ and $p_2 n$, respectively
    \STATE $\m{optim\_shuffle\_pass}(I, T_1, T_2, \pi_1, O)$
    \STATE $I \leftarrow O$
    \STATE $ \m{optim\_shuffle\_pass}(I, T_1, T_2, \pi, O)$
 \end{algorithmic}
\end{algorithm}

\begin{algorithm}[h]
\caption{\label{alg:optim-shuffle-pass}%
  Single pass $\mathsf{optim\_shuffle\_pass}(I, T_1, T_2, \rho, O)$ of
  the optimized \shuffle shuffle algorithm, where
  the user can read and store in private memory~$M$ up to $p_2 \sqrt{n}$~elements.}
$\bm{I}$: an array of $n$ encrypted elements $(x,v)$; $\bm{T_1}$: an
auxiliary array that fits $p_1 n$ elements; $\bm{T_2}$: an auxiliary
array that fits $p_2 n$ elements; $\bm{\rho}$: the desired permutation
for elements of~$I$; $\bm{O}$: the output array containing the
permutation of $I$ according to $\rho$, where every elements $(x,v)$
is re-encrypted.\\
\begin{algorithmic}[1]
    \STATE $\m{distr\_phase1}(I, \rho, T_1)$
    \STATE $\m{distr\_phase2}(T_1, \rho, T_2)$
    \STATE $\m{clean\_up\_phase}(T_2, \rho, O)$
  \end{algorithmic}
\end{algorithm}
}

\ifFull \CodeOptimShuffle \fi

The first distribution phase moves elements from $I$ to $T_1$, the
second distribution phase moves elements from $T_1$ to $T_2$.  We
abstract the layout of elements in each array further by sequentially
splitting buckets in chunks.  The goal of the first distribution phase
is to place elements in correct chunks and place them in correct
buckets within the chunks in the second distribution phase.
When a bucket is read, it is decrypted and any elements that are
written back are re-encrypted. In the following, we denote with
$\rho$ the target distribution of the shuffle pass.
For an array of size $n$, this algorithm assumes messages and
client private memory of size $O(\sqrt{n})$ and server
memory of size $O(n)$.

\paragraph{Distribution Phase \Rmnum{1}}
We view a sequence of $\sqrt[4]{n}$ buckets in each array as a chunk.
Hence, $I$, $T_1$, $T_2$ and $O$ each have $\sqrt[4]{n}$ chunks.  The
goal of the first distribution phase is to place elements of~$I$
in~$T_1$ in such a way that all elements that belong to the first
chunk of $O$ according to~$\rho$ can be found in the first chunk of
$T_1$, similarly for the second chunk, and so on.  \ifFull
The pseudo-code for
this phase is given in~Algorithm~\ref{alg:optim-distr-phase1}.
\else 
Details and pseudocode for this phase are given in
Appendix~\ref{app:optim-distr-phase1}.
\fi

\newcommand{\optdistI}{%
Referring to the pseudo-code of
Algorithm~\ref{alg:optim-distr-phase1}, the distribution proceeds by
placing elements in chunks at the bucket level
(line~\ref{line:optim-getRange-1}).
For every bucket in $I$, a batch with at most $p_1 \sqrt[4]{n}$
elements is moved to every chunk of $T_1$.  The assignment of elements
to chunks of~$T_1$ is determined by the permutation $\rho$ that we
wish to achieve.  In particular, an element $(x,v)$ is placed in the
$\lfloor\rho(x)/n^{3/4}\rfloor$th chunk of $T_1$
(line~\ref{line:optim-dist1-chunk}).  Note that
$\lfloor\rho(x)/n^{3/4}\rfloor$ is the index of the chunk where
$(x,v)$ belongs in $O$ since $n^{3/4}$ is the size of a chunk in~$O$.

As a result of placing $\sqrt{n}$ elements in $\sqrt[4]{n}$ batches of
size~$p_1\sqrt[4]{n}$, some batches will not be full. We pad such
batches with dummy elements (line~\ref{line:optim-dist1-pad}).  If, on
the other hand, more than $p_1 \sqrt[4]{n}$ elements need to be moved
according to $\rho$, the algorithm fails.  We consider and analyze
this case later.  If the distribution succeeds, every chunk of~$T_1$
has exactly $p_1 n^{3/4}$ elements with $n^{3/4}$ real elements only.
Note that here, as in the distribution phase of the basic shuffle,
$\sqrt[4]{n}$ calls in lines~\ref{line:optim-for-start11}-\ref{line:optim-for-end1}
can be substituted with a  single~$\putRangeDist$
of size~$p_1\sqrt{n}$.
For an illustration of this phase refer
to~Figure~\ref{fig:optim-distr-phase1}.

\ifFull
\begin{algorithm}[b!]
\else
\begin{algorithm}[hbtp]
\fi
\caption{\label{alg:optim-distr-phase1}%
  $\mathsf{distr\_phase1}(I, \rho, T_1)$ of the optimized \shuffle shuffle algorithm, where
  the user can read and store in private memory~$M$ upto $p_2 \sqrt{n}$~elements, $p_1 \le p_2$.}
$\bm{I}$: an array of $n$ encrypted elements $(x,v)$;
$\bm{\rho}$: the desired permutation;
$\bm{T_1}$: output array containing encrypted elements in correct chunks according to $\rho$.\\
\begin{algorithmic}[1]
    \STATE $\m{max\_elems} \leftarrow p_1  \sqrt[4]{n}$
    \STATE $\m{num\_chunks} \leftarrow \sqrt[4]{n}$
    \STATE $\m{num\_buckets} \leftarrow  \sqrt{n}$ 
    \STATE $\m{chunk\_size} \leftarrow n^{3/4}$
    \STATE \COMMENT{Distribution phase $\text{\Rmnum{1}}$: distribute elements of~$I$ in~$T_1$
    according to their chunk index in~$O$}
    \FOR[go through buckets in~$I$]{$i \in \{0, \ldots, \m{num\_buckets}-1\}$}
    	\STATE $\bucket_M \leftarrow \mathsf{getRange}(I,i \times \sqrt{n}, \sqrt{n})$ \label{line:optim-getRange-1}
	\STATE $\bucket[rev]_M \leftarrow \mathsf{empty\_map}()$ \COMMENT{Reverse map of chunk ids in $T_1$ to elements}~
	\FOR[Assign elements their chunk ids according to $\rho$]{$e \in \bucket_M$}
	         \STATE $(x,v) \leftarrow \Dec(e)$
	         \STATE $\mathsf{cid} \leftarrow  \lfloor \rho(x)/\m{chunk\_size} \rfloor$  \COMMENT{Chunk id of element $(x,v)$ in the output shuffle} \label{line:optim-dist1-chunk}
	         \STATE $\bucket[rev]_M[\m{cid}].\mathsf{add}(\Enc(x,v))$  \COMMENT{Collect elements that correspond to the same chunk in $T_1$}
	\ENDFOR
	\STATE \COMMENT {Can be done via a single $\putRangeDist$ for $\sqrt[4]{n}$ batches of size $\m{max\_elems}$}
	\FOR[Distribute $\bucket_M$ among chunks of~$T_1$]{$\mathsf{cid} \in \{0, \ldots,\m{num\_chunks}-1\}$}   \label{line:optim-for-start11}
		\IF{$\mathsf{size}(\bucket[rev]_M[\mathsf{cid}] ) > \mathsf{max\_elems}$}
			\STATE  \COMMENT{Permutation requires more than $p_1 \sqrt[4]{n}$ elements moved from a bucket of $I$ to a chunk of $T_1$}
			\STATE \textbf{fail} \label{line:distr1-fail}
		\ENDIF
		\STATE \COMMENT{Hide how many real elements go to $T_1$ from this bucket by padding with encrypted dummies}~
		\STATE $\bucket[rev]_M[\mathsf{cid}] \leftarrow \m{dummy\_pad}(\bucket[rev]_M[\mathsf{cid}], \m{max\_elems})$ \label{line:optim-dist1-pad}
		\STATE \COMMENT{Write a batch of $\mathsf{max\_elems}$ from every bucket of $I$ to every chunk of~$T_1$}~
		\STATE $\mathsf{putRange}(T_1,\m{cid} \times \m{chunk\_size} + i \times \mathsf{max\_elems}, \bucket[rev]_M[\m{cid}])$ \label{line:optim-putRange-1}
	\ENDFOR  \label{line:optim-for-end1}
    \ENDFOR
     \end{algorithmic}
\end{algorithm}

\begin{figure}[hbt]
\begin{center}
\includegraphics[scale=0.45]{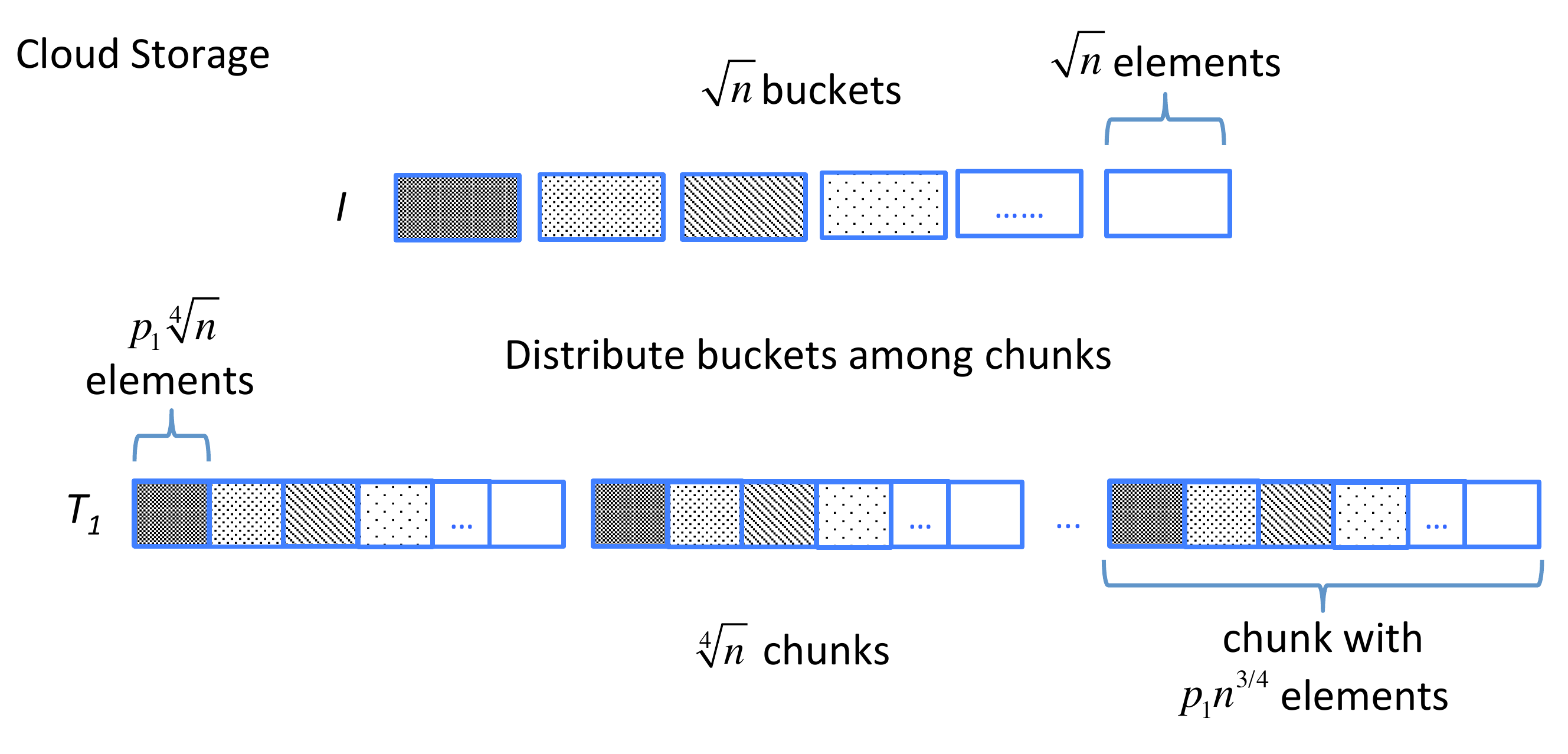}
\caption{Illustration of the arrangement of elements from the input~$I$
in the output~$T_1$ after the first distribution phase
of the Optimized Melbourne Shuffle~(Section~\ref{sec:optim_shuffle}).
See pseudocode in Algorithm~\ref{alg:optim-distr-phase1}.}
\label{fig:optim-distr-phase1}
\end{center}
\end{figure}
}

\ifFull \optdistI \fi
 
\paragraph{Distribution Phase \Rmnum{2}}
Observe that elements in $T_1$ belong to the correct chunk
but not the correct bucket within the chunk.
The second distribution phase remedies
this, such that by the end of this phase
elements of chunks of $T_1$
are in their correct buckets in $T_2$.
\ifFull
The pseudo-code of this phase is presented
in~Algorithm~\ref{alg:optim-distr-phase2}.
\else
Details and pseudocode for this phase are given in
Appendix~\ref{app:optim-distr-phase2}.
\fi

\newcommand{\optdistII}{%
Referring to the pseudo-code of Algorithm~\ref{alg:optim-distr-phase2}, 
we proceed by reading buckets of size $p_1\sqrt{n}$
in each chunk of~$T_1$ (line~\ref{line:optim-getRange-2}).
For every $j$th bucket of $i$th chunk of~$T_1$
we write
$\sqrt[4]{n}$ batches of size $p_2 \sqrt[4]{n}$ to~$T_2$~(line~\ref{line:optim-putRange-2}).
The batches are written to all
$\sqrt[4]{n}$ buckets of~$i$th~chunk in~$T_2$.
A real element $(x,v)$ is assigned to
$\lfloor \rho(x) / \sqrt{n}\rfloor \mod n^{3/4}$th batch~(line~\ref{line:optim-dist2-bid}),
which is the bucket id in the~$i$th chuck of $T_2$,
where this batch will be written to~(line~\ref{line:optim-putRange-2}).
Note that dummy elements added during the first distribution phase are ignored.
If there is a batch that no element has been assigned to
or if a batch has less than $p_2 \sqrt[4]{n}$
elements, we pad it with dummies.
If there is a batch with more than $p_2 \sqrt[4]{n}$
elements, the algorithm fails.
Note that $i \times \sqrt[4]{n} + \m{bid}$
is the index of $x$'s bucket in $O$.

Once all the batches are set up,
the batches are written in their place within
the same chunk in $T_2$.
Again,~$\sqrt[4]{n}$ calls in lines~\ref{line:optim-for-start2}-\ref{line:optim-for-end2}
can be substituted with a  single~$\putRangeDist$
of size~$p_2\sqrt{n}$.
When all buckets of $T_1$
have been processed, buckets of $T_2$
contain $p_2 \sqrt{n}$ elements each.
This is a consequence of writing
a batch of~$p_2 \sqrt[4]{n}$ elements
from every chunk of $T_1$.
Moreover, every bucket of $p_2 \sqrt{n}$
elements contains all $\sqrt{n}$ elements that belong to this
bucket in~$O$.
For an illustration of this phase refer to~Figure~\ref{fig:optim-distr-phase2}.

\begin{algorithm}[hbtp]
\caption{\label{alg:optim-distr-phase2}%
  $\mathsf{distr\_phase2}(T_1, \rho, T_2)$ of the optimized \shuffle shuffle algorithm, where
  the user can read and store in private memory~$M$ upto $p_2 \sqrt{n}$~elements.}
$\bm{T_1}$: array containing encrypted elements in correct chunks according to $\rho$;
$\bm{\rho}$: the desired permutation;
$\bm{T_2}$: output array containing encrypted elements in correct buckets according to $\rho$.\\
\begin{algorithmic}[1]
    \STATE $\m{max\_elems} \leftarrow p_2 \sqrt[4]{n}$
    \STATE $\m{num\_chunks} \leftarrow \sqrt[4]{n}$
    \STATE $\m{num\_buckets} \leftarrow  \sqrt{n}$ 
    \STATE $\m{bucket\_size} \leftarrow \sqrt{n}$
    \STATE $\m{chunk\_size} \leftarrow n^{3/4}$
    \STATE $\m{num\_buckets\_per\_chunk} \leftarrow \sqrt[4]{n}$
    \STATE \COMMENT{Distribution phase $\text{\Rmnum{2}}$: distribute elements of~$T_1$ in~$T_2$}
    \FOR[go through chunks in~$T_1$]{$i \in \{0, \ldots, \m{num\_chunks}-1\}$}
    \FOR[go through buckets in~$i$th chunk]{$j \in \{0, \ldots, \m{num\_buckets\_per\_chunk} -1\}$}
    	\STATE $\bucket_M \leftarrow \mathsf{getRange}(I,i\times \m{chunk\_size} + j \times \sqrt{n}, p_1\sqrt{n})$ \label{line:optim-getRange-2}
	\STATE $\bucket[rev]_M \leftarrow \mathsf{empty\_map}()$ \COMMENT{Reverse map of bucket ids in $T_2$ to elements}~
	\FOR[Assign elements their bucket ids according to $\rho$]{$e \in \bucket_M$}
	         \STATE $(x,v) \leftarrow \Dec(e)$
	         \IF[Ignore dummy elements]{$(x,v)$ is real}	         
	            \STATE $\mathsf{bid} \leftarrow  \lfloor \rho(x) / \m{bucket\_size} \rfloor \mod \m{num\_buckets\_per\_chunk}$  \COMMENT{Bucket id of element $(x,v)$ in $i$th chunk of $T_2$} \label{line:optim-dist2-bid}
	         \STATE $\bucket[rev]_M[\m{bid}].\mathsf{add}(\Enc(x,v))$  \COMMENT{Collect elements that correspond to the same bucket}
	        \ENDIF
	\ENDFOR
	\STATE \COMMENT {Can be done via a single $\putRangeDist$ for $\sqrt[4]{n}$ batches of size $\m{max\_elems}$}
	\FOR[Distribute $\bucket_M$ among buckets of~$T_2$ in $i$th chunk)]{$\mathsf{bid} \in \{0, \ldots,\m{num\_buckets\_per\_chunk}-1\}$} \label{line:optim-for-start2}
		\IF{$\mathsf{size}(\bucket[rev]_M[\mathsf{bid}] ) > \mathsf{max\_elems}$}
			\STATE  \COMMENT{Permutation requires more than $p_2 \sqrt[4]{n}$ elements moved from $T_1$ to $T_2$}
			\STATE \textbf{fail} \label{code:distr2-fail}
		\ENDIF
		\STATE \COMMENT{Hide how many real elements go to $T_2$ from this bucket by padding with encrypted dummies}~
		\STATE $\bucket[rev]_M[\m{bid}] \leftarrow \m{dummy\_pad}(\bucket[rev]_M[\mathsf{bid}], \m{max\_elems})$
		\STATE \COMMENT{Write a batch of $\mathsf{max\_elems}$ from every bucket of $T_1$ to buckets
		in $i$th chunk in~$T_2$}~
		\STATE $\mathsf{putRange}(T_2,i \times \m{chunk\_size} + \m{bid} \times p_2 \sqrt{n} 
		+ j\times\m{max\_elems}, \bucket[rev]_M[\m{bid}])$ \label{line:optim-putRange-2}
	\ENDFOR \label{line:optim-for-end2}
    \ENDFOR
    \ENDFOR
     \end{algorithmic}
\end{algorithm}

\begin{figure}[hbt]
\begin{center}
\includegraphics[scale=0.45]{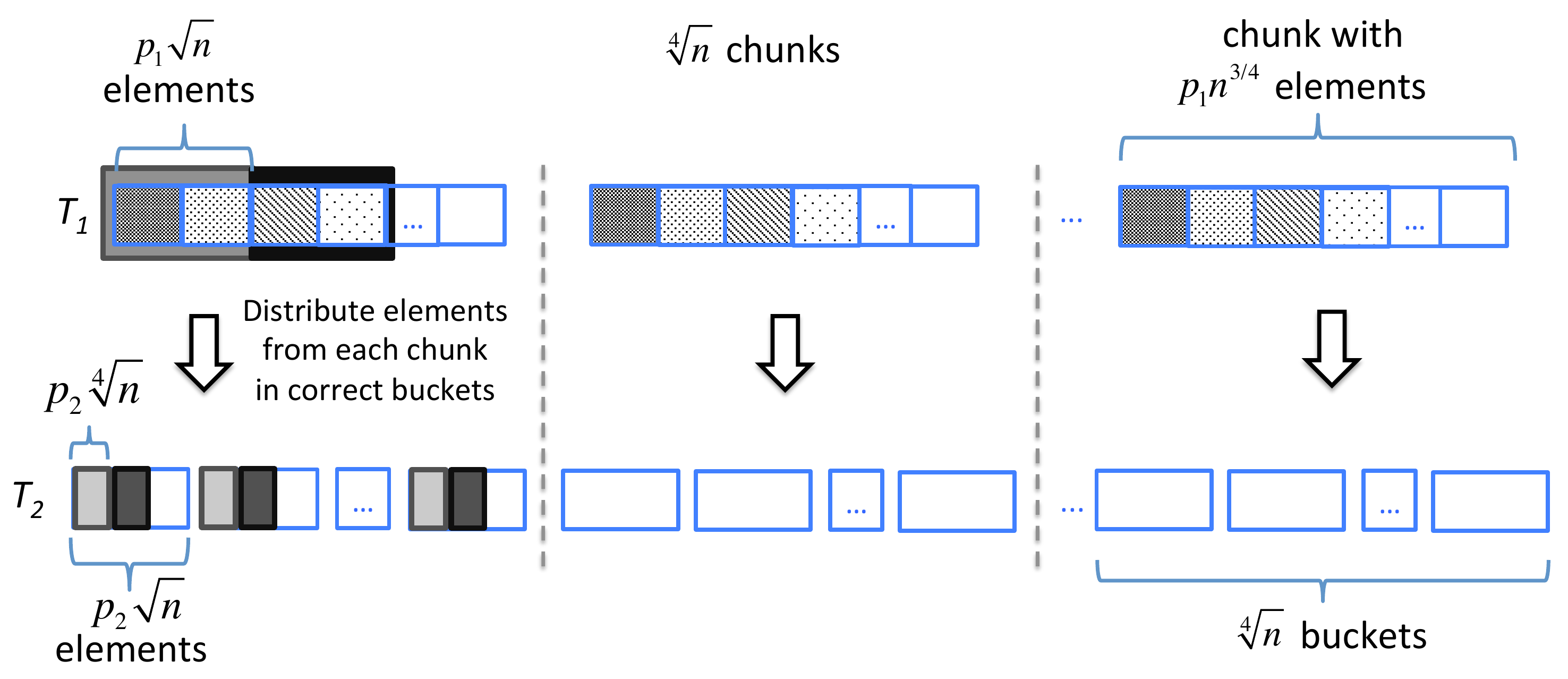}
\caption{Illustration of the arrangement of elements 
after processing two buckets of the first chunk of the input~$T_1$
in the output~$T_2$ in the second distribution phase
of the optimized \shuffle Shuffle. See pseudocode in Algorithm~\ref{alg:optim-distr-phase2}.
Read and written elements are indicated in gray and black colors.}
\label{fig:optim-distr-phase2}
\end{center}
\end{figure}
}

\ifFull \optdistII \fi

\paragraph{Clean-up Phase}
This phase is similar to the clean-up
phase of the basic version,
hence we omit its pseudo-code.
\ifFull Recall that a bucket of $T_2$
contains
$p_2 \sqrt{n}$ encrypted elements
but contains only $\sqrt{n}$
real elements which are in the correct bucket
but not in the correct spot within the bucket.
\else
The elements are in the correct buckets
but not in the correct spots.
\fi
We remedy this by reading every bucket $j\in [1, \sqrt{n}]$,
decrypting it,
removing dummy elements such that only
$\sqrt{n}$ real elements are left,
sorting it according to~$\rho$,
re-encrypting the elements and writing them back
to the $j$th bucket of~$O$.

\paragraph{Performance}
The performance of the optimized \shuffle shuffle is summarized in the
following theorem.
\proofInApp{app:optim-shuffle}
\begin{theorem} \label{theorem:optim-shuffle}
Given an input array of size $n$, 
the optimized Melbourne shuffle (Algorithm~\ref{alg:optim-shuffle-alg})
executes~$O(\sqrt{n})$ operations, each exchanging a message of size
$O(\sqrt{n})$ between
a user with private memory
of size~$O(\sqrt{n})$ and a server
with storage of size~$O(n)$.
Also, the user and the server perform~$O(n)$ work.
\end{theorem}
\ifFull \else {The \shuffle shuffle is highly parallelizable
(see Appendix~\ref{app:pram} for details).\fi
\newcommand{\ProofOptimShuffle}{%
\begin{proof}
A single optimized shuffle pass in Algorithm~\ref{alg:optim-shuffle-pass}
requires $3 \sqrt{n}$ calls to $\m{getRange}$,
$2 \sqrt{n}$ calls to $\m{putRangeDist}$
and $\sqrt{n}$ calls to $\m{putRange}$
when the user and the server can exchange up to $p_2 \sqrt{n}$
elements in a single request.
Then, the optimized shuffle algorithm
in Algorithm~\ref{alg:optim-shuffle-alg}
requires $12$ requests in total.
The size of the user private memory required to perform the shuffle is at most
$p_2 \sqrt{n}$,
while
required server's memory is at most $(p_1 + p_2) n$.
However, this overhead in server's memory is temporary
since this memory is required only during the
the distribution phases. Similarly, for the user
private memory, the size of $O(\sqrt{n})$
is required while shuffling and is decreased to $O(1)$
once it is finished.
The total work for the user is $O(n)$.
\ifLCNS \qed \fi
\end{proof}
}
\ifFull \ProofOptimShuffle \fi

\ifFull
\subsection{Security Analysis}
\else
\paragraph{Security Analysis}
\fi

\ifFull
In this section, we show \else
We show in Appendix~\ref{app:security:optimized}
\fi
that the optimized \shuffle shuffle is
oblivious for every permutation $\pi$ with very high probability.

\newcommand{\SecOptimShuffle}{%
\begin{definition}
\label{def:q-perm}
Let $A$ be an array of $n$ elements such that for every $1 \le x \le n$ is
at location~$\pi_0(x)$ in~$A$.
Let $B=\pi(A)$ be a permutation on $A$.
Split~$A$ and $B$ in $\sqrt[4]{n}$ chunks of equal size
and fix constants $p_1,p_2 \ge 1$.
Split every chunk of $A$ and $B$ further
in $\sqrt[4]{n}$ buckets of size $\sqrt{n}$.

We say that a permutation $\pi \in Q(\pi_0)$
if according to $\pi$: (1) every chunk of $B$
has at most $p_1 \sqrt[4]{n}$
elements from every bucket of $A$,
and (2) if every bucket of~$B$
has at most $p_2 \sqrt[4]{n}$
elements from every chunk of~$A$.
\end{definition}

\begin{lemma}
The size of $Q(\pi_0)$ is
$(1 - \m{negl}(n)) \times n!$ for every permutation~$\pi_0$.
\end{lemma}

\begin{proof}
Let $A$ and $B$ be arrays of $n$
items. Split $A$ and $B$ in chunks
of size $n^{3/4}$
and then split each chunk
in $\sqrt[4]{n}$ buckets of size $\sqrt{n}$.
The shuffle pass succeeds
for any permutation that satisfies the following
two constraints.
(1) Every chunk of $B$ contains
at most $p_1\sqrt[4]{n}$ elements
from every bucket of $A$
and (2)
Each bucket of $B$
has at most $p_2 \sqrt[4]{n}$
elements of the $i$th chunk of $A$.
We compute the probability
of a uniform permutation
not satisfying each of the constrains
independently and use union bound
to bound the probability of
failing at least one of the constraints.

Let $X_a^b$ be the
number of elements of
$a$th bucket of $A$ in $b$th chunk of $B$.
Since there are $\sqrt{n}$
elements in every bucket of A
and there are $\sqrt[4]{n}$
chunks in~$B$,
the mean of $X_a^b$ is $\sqrt[4]{n}$. 
As in Lemma~\ref{lemma:bb-perms}
we use the Poisson approximation
and use independent
Poisson variables $Y_a^b$
with mean value of $\sqrt[4]{n}$.
Then for a specific $a$ and $b$:
\[
\Pr[Y_a^b \ge p_1\sqrt[4]{n}] \le \frac{{\me}^{-\sqrt[4]{n}} (\me {\sqrt[4]{n}})^{p_1{\sqrt[4]{n}}}}{ {(p_1\sqrt[4]{n}})^{p_1\sqrt[4]{n}} } = 
 \frac{{\me}^{p_1} {\sqrt[4]{n}}^{p_1{\sqrt[4]{n}}}}{ {(p_1\sqrt[4]{n}})^{p_1\sqrt[4]{n}} }=
 \frac{{\me}^{p_1}}{ {p_1}^{p_1\sqrt[4]{n}} }. 
\]

Using union bound the probability:
\[
\Pr[\bigcup_{\substack{\text{$1 \le a \le \sqrt{n}$}\\\text{$1 \le b \le \sqrt[4]{n}$}}} Y_a^b \ge p_1\sqrt[4]{n}]
\le n^{3/4} \frac{{\me}^{p_1}}{ {p_1}^{p_1\sqrt[4]{n}} } = \m{negl}[n].
\]
Hence, 
\[
\Pr[\medcup{X_a^b} \ge p_1\sqrt[4]{n}] \le 
2\Pr[\medcup{Y_a^b} \ge p_1\sqrt[4]{n}]
= \m{negl}[n].
\]

Consider the second constraint.
Let $U_c^d$ be the the number of elements of $c$th chunk
of $A$ in $d$th bucket of~$B$. 
Since there are $n^{3/4}$ elements in every
chunk of $A$ and $B$ has $\sqrt{n}$ buckets,
$U_c^d$ is a Poisson random
variable with parameter $\sqrt[4]{n}$.
Using the Poisson approximation,
let $W_c^d$ be $n^{3/4}$ independent
Poisson random variables with mean~$\sqrt[4]{n}$.
Similar to the analysis of the first constraint,
the probability that there is at least one bucket
of~$B$ that has more than~$p_2 \sqrt[4]{n}$
elements of a chunk of $A$ is $\m{negl}(n)$.

Using union bound the probability
that at least one of the constraints is not
satisfied is the sum of
the respective probabilities, hence,
with probability $1-\m{neg}(n)$,
moving elements for
a random permutations via two
levels will succeed.
\ifLCNS \qed \fi
\end{proof}

\begin{lemma}
\label{lemma:data-indep-optim}
The metadata of requests exchanged between
the client and the server in method $\optimshufflepass$
(Algorithm~\ref{alg:optim-shuffle-pass})
is independent of input permutation~$\pi_0$
and output permutation~$\rho \in Q(\pi_0)$,
and depends only on~$n$.
\end{lemma}

\begin{proof}[Sketch]
We consider the metadata from $\getRange$ and $\putRange$
requests during the two distribution phases and the clean-up phase.
During the distribution phase~\Rmnum{1} $\sqrt{n}$
$\getRange$ calls are made with metadata
representing sequential get accesses
to blocks of size $\sqrt{n}$ of input~$I$~(line~\ref{line:optim-getRange-1}
in Algorithm~\ref{alg:optim-distr-phase1}).
The call to $\putRangeDist$ puts blocks of data of size~$p_1 \sqrt{n}$
to locations in~$T_1$ that are determined based on the id of the block
being read from~$I$ (lines~\ref{line:optim-for-start11}--\ref{line:optim-for-end1}
in Algorithm~\ref{alg:optim-distr-phase1}). Hence, metadata depend on the
size of~$I$, $n$, and does not depend on the content.

During the distribution phase~\Rmnum{2}, blocks of size $p_1\sqrt{n}$
are sequentially read using $\getRange$ on~$T_1$~(line~\ref{line:optim-getRange-2}
in Algorithm~\ref{alg:optim-distr-phase2})
and block are written back using~$\putRangeDist$ to $T_2$
of size~$p_2 \sqrt{n}$~(lines~\ref{line:optim-for-start2}--\ref{line:optim-for-end2}
in Algorithm~\ref{alg:optim-distr-phase2}).
Hence, metadata is deterministic
based on~$n$.

Clean-up phase sequentially reads blocks of size~$p_2\sqrt{n}$
from~$T_2$
and writes back blocks of size~$\sqrt{n}$, independent of the data.
Hence, metadata, produced by Algorithm~\ref{alg:optim-shuffle-pass}
is data independent.
\ifLCNS \qed \fi
\end{proof}

\begin{lemma}
\label{lemma:correctness-optim}
Let $\pi_0$ be the initial permutation of $n$ elements
in~the input array~$I$.
Algorithm~\ref{alg:optim-shuffle-pass} succeeds for
all permutations $\rho \in Q(\pi_0)$ and
is data-oblivious according to Definition~\ref{def:obl-shuffle}.
\end{lemma}
\begin{proof}[Sketch]
The algorithm can fail during the first or second distribution
phases (lines~\ref{line:distr1-fail} in Algorithm~\ref{alg:optim-distr-phase1}
and~\ref{code:distr2-fail} in Algorithm~\ref{alg:optim-distr-phase2}).
The first failure happens when distributing
a bucket of size~$\sqrt{n}$ of~$I$ among~$\sqrt[4]{n}$ buckets in$~T_1$
and there is a bucket that needs to put
more than $p_1\sqrt{n}$ of its elements to
one of the buckets according to the input
permutation~$\rho$. This failure represents failure
to satisfy the first condition of~Definition~\ref{def:q-perm} of~$Q(\pi_0)$.

The second distribution phase results in a failure when one of
the buckets in~$T_1$, formed in the previous step,
is distributed among~$\sqrt[4]{n}$ buckets of~$T_2$
and requires more than~$p_2 \sqrt[4]{n}$ of its elements
in one of the buckets of~$T_2$.
This corresponds to failing the second condition of~Definition~\ref{def:q-perm}.
Hence, Algorithm succeeds for all permutations in~$Q(\pi_0)$.

The metadata produced by the algorithm is data-independent
as showed in Lemma~\ref{lemma:data-indep-optim}.
Since data is re-encrypted every time it is written back
in $\getRange$ we can use the reduction similar to Lemma~\ref{lemma:shuffle-pass}
to show that the security of the shuffle depends on
Enc-IND-CPA secure encryption scheme.
\ifLCNS \qed \fi
\end{proof}

\begin{theorem}
\label{thm:optim-shuffl}
The optimized Melbourne Shuffle (Algorithm~\ref{alg:optim-shuffle-alg})
is a randomized shuffle algorithm
that succeeds with very high probability
and is data-oblivious according to Definition~\ref{def:obl-shuffle}.
\end{theorem}
\begin{proof}[Sketch]
Algorithm~\ref{alg:optim-shuffle-alg} runs the optimized shuffle
pass algorithm in Algorithm~\ref{alg:optim-shuffle-pass} twice.
Each shuffle pass succeeds if the desired permutation
is from the set $Q(\pi_0)$, where $\pi_0$ is the original
permutation of the input~$I$.
Running the shuffle twice: once for a random
permutation~$\pi_1$ and then for the desired permutation,
ensures that the algorithm can succeed for every
permutation with very high probability, where
the probability depends on random coin tosses
that determine the intermediate random permutation~$\pi_1$.
We showed in Lemma~\ref{lemma:correctness-optim}
that Algorithm~\ref{alg:optim-shuffle-pass}
is data-oblivious, hence, running it twice also
produces a data-oblivious algorithm (same reduction argument
as in Theorem~\ref{thm:shuffle}).
\ifLCNS \qed \fi
\end{proof}

} 

\ifFull \SecOptimShuffle \fi


\section{The Melbourne Shuffle with Small Messages}
\label{sec:small-msg-shuffle}

The \shuffle shuffle
and its optimized version can be extended
to work with messages and private memory
of size  $n^{1/c} \log n$ (or $n^{1/c}$ for the optimized version),
for $c \ge 3$.
\ifLCNS The algorithm is described in
Appendix~\ref{app:small-msg-shuffle}. \fi

\newcommand{\SmallMsgShuffle}{%

The idea behind the approach is to
run the algorithm recursively with depth~$c-1$.
For a fixed~$c$,
one first splits the output in large buckets
of size $n^{(c-1)/c}$ and executes the
shuffle as in the square root case:
distributing~$n^{1/c}$ among
$n^{1/c}$ large buckets. We call
this the first level of the recursion.
After this, each large bucket has correct elements
but not in correct buckets nor positions.
The square root shuffle is executed again
on each large bucket, but now splitting
the large bucket of size $O(n^{(c-1)/c} \log n)$~($O(n^{(c-1)/c})$ for the optimized version)
in~$n^{(c-2)/c}$ buckets and again using
only~$n^{1/c}$ private memory.
The client follows the recursion
until the size of the inner buckets becomes $O(n^{2/c})$
when elements can be distributed in their correct
buckets of size~$n^{1/c}$.
At this point, the buckets are small enough
that they can be read to private memory
during the clean-up phase and be placed
in correct positions within their bucket.
Hence, there are $c-1$ levels of recursion.
Each level~$i$ requires a block of $n^{(c-i)/c}$ buckets,
each of size~$n^{1/c}$,
to be distributed among $n^{1/c}$ output buckets.
Since every level has $n^{(i-1)/c}$ blocks,
the square root solution is required
to be executed~$n^{(i-1)/c}$~times per level,
each making~$O(n^{(c-i)/c})$
accesses. Hence, the total number of requests
can be expressed
as
$\sum_{i=1}^{c-1} O(n^{(c-i)/c} \times n^{(i-1)/c}) = O((c-1) n^{(c-1)/c})$.

For the Melbourne shuffle that uses
private memory and messages with a multiplicative $\log n$
factor (Section~\ref{sec:shuffle}), the naive solution
could accumulate the $O({(\log  n)}^c)$ factor
if used naively. This happens due
to reading $n^{1/c}$ elements and writing back~$O(n^{1/c} \log n)$
elements and clean-up phase not being able to reduce this
to $n^{1/c}$ until the last level of the recursion.
One can prevent this by observing that
when distributing $n^{2/c}$ elements
among $n^{1/c}$ output buckets,
every bucket 
will have at most $n^{1/c}$ of its elements
in every output bucket, with very high probability
(by using Chernoff bounds~\cite[Chapter 4]{mu-pcrap-05}).
Hence, after distributing $O(\log n)$
elements from buckets of size
$n^{1/c}$, we can make another sequential pass,
reading $n^{1/c} \log n$ elements at a time
(i.e., the elements that were contributed by
batches of size~$\log n$ from~$n^{1/c}$
buckets)
and applying the fact that all together
they could not contribute more than $O(n^{1/c})$
elements, and hence writing back only~$O(n^{1/c})$.
Note that we remain data-oblivious, since
again we are using the observation of what the
adversary would expect to see with very high probability.

We note that each recursive level does not depend
on higher levels and, hence, has the
same memory requirements as a single shuffle
pass of the corresponding algorithm. 
}

\ifFull \SmallMsgShuffle \fi

\begin{theorem} \label{theorem:small-msg-shuffle} %
  Given an integer constant $c \geq 3$ and
  an input array of size $n$, 
  the optimized \shuffle shuffle executes
  $O(cn ^{(c-1)/c})$ operations, each exchanging a message of size
  $O(n^{1/c})$ between a user with private memory of size
  $O(n^{1/c})$ and a server with storage of size $O(n)$. Also, the
  user and server perform $O(c n)$ work.
\end{theorem}


\ifFull
\section{Applications}

In this section, we show how the \shuffle shuffle can be efficiently
parallelized on a PRAM. We also show how to use the  \shuffle shuffle
to obtain an efficient oblivious storage solution.

\else
\section{Oblivious Storage}
\label{sec:oblivious:storage}
\fi

\newcommand{\PRAMShuffle}{
In the EREW PRAM model, $\sqrt{n}$ processors can agree on
a seed for a permutation of the shuffle and
run the shuffle as specified below.
If privacy is not an issue, encryption and decryption
calls can be omitted from the algorithms
to get a data independent shuffle algorithm.

\begin{theorem}
  The (optimized) Melbourne shuffle can be executed in $O(1)$ steps by
  $\sqrt{n}$ processors, each with $O(\sqrt{n}\log n)$ ($O(\sqrt{n})$)
  private memory, that access a shared memory of size
  $O(n \log n)$ ($O(n)$) via EREW protocol.
\end{theorem}
\begin{proof}
(Sketch)
Label each processor with an id from 0 to $\sqrt{n}-1$.
During the shuffle pass in~Algorithm~\ref{alg:shuffle-pass}
every bucket~$\m{id}_I$ of $I$ can be read by
the $\m{id}_I$th processor from the shared memory.
The processor creates batches of elements that go
from his bucket to output buckets in the shared memory.
He then can write the batches
to the locations of the shared memory of~$O$
as specified in line~\ref{line:simple-putRange1} of Algorithm~\ref{alg:shuffle-pass}
by using processor's id instead of $\m{id}_I$.
During the clean-up phase each processor can read the bucket
of~$O$ with the same id as his processor id and write it back.
(Changes to Algorithm~\ref{alg:optim-shuffle-pass} are analogous.)
\ifLCNS \qed \fi
\end{proof}
}
\ifFull
\subsection{PRAM}
 \PRAMShuffle

\subsection{Oblivious Storage}

\label{sec:oblivious:storage}
 \fi

In this section, we give an overview of
a secure and efficient oblivious storage scheme
that uses the Melbourne shuffle.
The oblivious storage (OS) we
consider here follows the framework proposed
in~\cite{go-spsor-96} and the follow-up work of~\cite{prac_oram}.
The goal of the oblivious storage is to hide
client's access pattern to his remotely stored data from anyone
observing it, including the storage provider.
Informally, OS transforms a virtual
sequence of requests into a simulated one
that appears to be data-independent.
This is achieved by a mixture of
accesses that are the same for every
access sequence (e.g., the Melbourne shuffle)
and of accesses that are randomized, and come from
the same distribution, hence they appear to
be independent of
the access sequence.

\paragraph{OS Scheme}

Our OS scheme consists of \emph{setup},
\emph{access} and \emph{rebuild}
phases. The setup phase
arranges, encrypts
and outsources the data to the
remote storage server.
The access phase transforms a virtual
request into a sequence of accesses
to the remote storage. Once these
accesses are performed,
the requested element
is returned.
After a batch of requests, the
data at the server is shuffled
in order to be able to proceed 
with the access phase for
the next batch.
\ifFull 

We \else In Appendix~\ref{app:oram}, we \fi
 provide a short description
of the \emph{square root OS solution}~\cite{go-spsor-96}
and show how the Melbourne Shuffle
improves its performance.

\newcommand{\ORAMShuffle}{%

In the following description we assume
that the user and the server
can exchange messages of
size $O(\sqrt{n})$ and the client's memory
is also $O(\sqrt{n})$.

\paragraph{Setup}
Let $A$ be an array on $n$ items.
The client extends $A$ by
adding $\sqrt{n}$ fake
elements with keys $n+1, n+2, \ldots, n + \sqrt{n}$.
He then encrypts~$A$ using Enc-IND-CPA secure
encryption
to get an array~$I$, which he sends to the server
to store.
The client picks a secret permutation, PRP~$\pi$
and calls $\m{shuffle}(I, \pi, O)$
where $O$ is the location at the server
where the shuffled and encrypted
array of $A$ is stored.
We set $I \leftarrow O$ for the access phase.
The user also allocates at the server an empty cache $C$
that can fit encryptions of up to
$\sqrt{n}$
requested elements.
As we shall see, the number of non-empty locations
in~$C$ is the total number of requests
that were made to remote storage since
the last time $I$ was shuffled.
We refer to how many elements
are present in $C$ as $l$.
The client remembers the seed that generated~$\pi$
so that he can find the elements during
the access phase.

\paragraph{Access Phase}
Given the cache $C$, the encrypted array~$I$ and a request $\m{read}(x)$ or $\m{write}(x,v')$
the access phase creates the following
oblivious simulation
that depends on $n$, the total number of accesses made so far,
and secret PRP~$\pi$.
It starts by reading the cache with
a single request since the cache fits
in one request message and in client's memory.
It decrypts the cache and
checks if an element with the key~$x$ is present in~$C$.
If so, it remembers the element $(x,v)$.
Then, a location in $I$ is accessed as follow.
If the element was found in the cache
a fake element with the key~$n+l$
is accessed by requesting the location
$\pi(n+l)$ of~$I$.
Otherwise, the location that stores
the element with the key~$x$ is accessed,
by requesting location~$\pi(x)$ of~$I$.
After reading the cache and making one request to~$I$,
the user has the desired element $(x,v)$. If the original
request was $\m{read}$ he writes encrypted element $(x,v)$
to the first empty location in~$C$ on the server,
if it was a $\m{write}$, the client writes $(x,v')$ instead. 
This phase can proceed this way for $\sqrt{n}-1$ more
requests, after that the cache fills up and the
rebuild phase follows.
This phase requires 3 accesses to the remote
storage per every requested element.
\paragraph{Rebuild Phase}
The goal of the rebuild phase is to free $C$
by placing
updated elements back to~$I$ and
shuffle~$I$ using a new secret permutation~$\pi'$
that is independent of $\pi$. This step has to be
done in a data-oblivious manner to prevent
correlations between access patterns to~$I$ before
and after reshuffle.
One first writes~$C$ to~$I$ by reading~$C$
in private memory and then reading
buckets of size~$\sqrt{n}$ of~$I$,
updating them with elements of~$C$, if needed,
and writing them back re-encrypted.
After merging the cache~$C$ and~$I$ we
are ready to call the Melbourne shuffle via~$\m{shuffle}(I, \pi', O)$.
The client updates the seed that generated~$\pi'$
in his private memory and allocates an empty cache~$C$
at the server.
The rebuild takes $O(\sqrt{n})$
accesses if we use the optimized Melbourne Shuffle.

\paragraph{Deamortized Oblivious Storage}
The overall cost of the shuffle
can be amortized over a batch of~$\sqrt{n}$
requests to achieve $O(1)$ overhead.
We could also deamortize this by using the method
of~\cite{gmot-orsew-11}. The method involves
doubling the space at the server by adding a second cache of size~$\sqrt{n}$
and memory of size $n+\sqrt{n}$ where the rebuild
is happening ``behind the scenes'' during the access phase.
The rebuild on the
new array is done by batching its requests
with the requests from the access phase.
The client now does a constant increase
in the amount of work and has a new permutation
ready
when the cache is full.
}
\ifFull \ORAMShuffle \fi
%
%
\ifFull
The performance of the above OS scheme base on the optimized \shuffle
shuffle is summarized in the following theorem.
\fi
\newpage
\begin{theorem}  \label{theorem:oram:shuffle}
The randomized oblivious storage scheme based on the optimized \shuffle shuffle has the
following properties, where $n$ is the size of the outsourced dataset:
\ifLCNS
\begin{itemize}[noitemsep]
\else
\begin{itemize}
\fi
\item The private memory at the client and each message exchanged
  between the client and server have size~$O(\sqrt{n})$.
\item The memory at the server has size $O(n)$.
\item The access overhead to perform a storage request is $O(1)$.
\end{itemize}
\end{theorem}

\paragraph{Extension to Small Messages}
The recursion could be applied to square root solution,
when messages of size $n^{1/c}$ are used
to exchange between the user and the server,
and the user has $n^{1/c}$ private memory.
The solution contains
one
cache of size~$n^{1/c}$, that is small enough
to fit into private memory and be read using one request,
and~$c-1$~levels. 
Each level $i$ is large enough to
contain~$n^{(i+1)/c}$ real elements
and~$n^{i/c}$ fake elements.
The cache and levels 1 to $i-1$ have
a similar cache functionality for level $i$ as the cache~$C$ in
the square root solution, except together
they can store $O(n^{i/c})$ previously accessed elements. 
Each level~$i< c-1$ contains $n^{(i+1)/c}$ buckets of size $O(\log n)$
that allows one to store $n^{(i+1)/c}$ elements
in a hash table and avoid collisions with very high probability.
Buckets that have less than $\log n$ elements\ifFull, real or fake,\else~\fi
are filled in with dummies.
The last level, level $c-1$, has $n$ elements,
hence a permutation can be used to store and
access the elements.
\newcommand{\ORAMShuffleSmall}{%

Given the memory arrangement of the
extension with small messages at the server from~Section~\ref{sec:oblivious:storage},
the access and rebuild phases
proceed as follows.

\emph{Access phase} requires a total of $c+1$ accesses
to the server: read the cache of size~$n^{1/c}$,
read $O(\log n)$ size bucket from every $c-2$ levels,
read one element from the last level, write an updated
cache back. Note that each bucket
can be read into memory since we assume
private memory and messages of size~$O(n^{1/c})$.
Hence,~$O(c)$ accesses are required to
access an element obliviously.
The access phase proceeds as follows:
read the cache, if an element is found access
a new fake element from level 1, and proceed
with fake accesses from then on.
Otherwise, look up the bucket using a hash function of level 1
where the element is supposed to be if it was read before.
Again, proceed with consequent fake accesses,
if the element was found, or keep looking for the element.

\emph{Rebuild phase:} After $n^{1/c}$ elements
have been accessed, level 1 is rebuilt taking
$O(n^{1/c} \log n)$ accesses to be rebuilt using the
Melbourne Shuffle for small messages~(see Section~\ref{sec:small-msg-shuffle}). 
Similarly, when ($i$-1)th level is full it requires
shuffling of $n^{(i+1)/c} \log n$ elements at level~$i$
using $O(i \times n^{i/c} \log n)$ accesses. Finally, the last
level requires~$O((c-1)n^{(c-1)/c})$ accesses to rebuild.
We could amortize the cost of the rebuilds to get
$O(c \log n)$ amortized, or deamortized
using~\cite{gmot-orsew-11}, access overhead per every element.
We could also increase message size to~$O(\sqrt[c]{n} \log n)$
and use the optimized Melbourne shuffle to achieve
a constant overhead from the rebuild, since
now the rebuild for level~$i$ takes
 $O(n^{i/c})$ accesses and can be (de)amortized over
 $O(n^{i/c})$ elements that caused it.
 Hence, with messages of size $\sqrt[c]{n} \log n$,
 and using the optimized Melbourne shuffle
 we get $O(c)$ (de)amortized access overhead.

Though, this solution resembles the Hierarchical
Solution of~\cite{go-spsor-96} by using buckets of size~$\log n$
at every level, it has two important differences.
The expansion factor from level to level is $n^{1/c}$,
instead of $2$, hence, we only have $c$ levels,
and the rebuild phase uses our shuffle
algorithm instead of a more expensive oblivious sort.
}
\ifFull 

\ORAMShuffleSmall \else
Details are given in
Appendix~\ref{app:oram:small}.
\fi

\begin{theorem}
  \label{theorem:oram:shuffle:small}
The randomized oblivious storage scheme based on the \shuffle shuffle
with small messages has the
following properties, where $n$ is the size of the outsourced dataset
and $c$ is a constant such that $c \geq 3$:
\ifLCNS
\begin{itemize}[noitemsep]
\else
\begin{itemize}
\fi
\item The private memory at the client and each message exchanged
  between the client and server have size $O(\sqrt[c]{n})$ $(O(\sqrt[c]{n} \log n))$.
\item The memory at the server has size $O(n)$.
\item The access overhead to perform a storage request is $O(c \log n)$ $(O(c))$.
\end{itemize}
\end{theorem}


\ifFull
\vspace{-8pt}
\section*{Acknowledgements}
  This research was supported in part by
  the National Science Foundation under grants CNS--1011840,
  CNS--1012060, CNS--1228485, CNS--1228639, and IIS--124758,
  by the National Institutes of Health under grant R01-CA180776, and
  by the Office of Naval Research under grant N00014-08-1-1015.
  Olga Ohrimenko worked on this project in part while at Brown
  University.
  %

\fi


\ifLCNS
\bibliographystyle{splncsnat}
\else
\bibliographystyle{abbrv}
\fi
\bibliography{paper}

\ifLCNS
\clearpage

\input{appendix}

\fi

\end{document}